\newif\iflong
\longtrue

\documentclass[acmsmall]{acmart}

\setcopyright{rightsretained}
\acmJournal{PACMPL}
\acmYear{2021} \acmVolume{5} \acmNumber{POPL} \acmArticle{15} \acmMonth{1} \acmPrice{}\acmDOI{10.1145/3434296}
\usepackage{listings}
\lstdefinelanguage{program}{%
  keywords={%
    let,pass,function,%
    var,const,bool,int,void,atomic,%
    while,do,if,then,else,assume,assert,call,return,rule,forall,with,new,choose,skip,%
    task,async,yield,for,wait,%
    type,relation,init, action, safety, invariant, axiom, input,repeat
  },
  morecomment=[l]{//},
  morecomment=[s]{/*}{*/},
  morecomment=[n]{(*}{*)},
  mathescape=true,
  escapeinside=`',
}
\usepackage{booktabs}
\usepackage{mathpartir}
\usepackage{multirow}
\usepackage{amsmath}
\usepackage{amsthm}
\usepackage{stmaryrd}
\usepackage{array}
\usepackage{subcaption}
\usepackage[T1]{fontenc}
\usepackage[nodayofweek]{datetime}
\usepackage{graphicx}
\usepackage{algorithm}
\usepackage{wrapfig}
\usepackage[noend]{algpseudocode}
\usepackage[flushleft]{threeparttable}
\usepackage{stmaryrd}
\usepackage{multicol}
\usepackage{pifont}
\usepackage[xcolor]{changebar}
\usepackage{caption}
\cbcolor{magenta}

\usepackage{xcolor}
\usepackage{extarrows}
\usepackage{mathtools}
\usepackage{paralist}

\usepackage{tikz}
\usetikzlibrary{positioning,fit}

\usepackage[noabbrev,capitalise]{cleveref} %

\makeatletter
\newcounter{algorithmicH}%
\let\oldalgorithmic\algorithmic
\renewcommand{\algorithmic}{%
  \stepcounter{algorithmicH}%
  \oldalgorithmic}%
\renewcommand{\theHALG@line}{ALG@line.\thealgorithmicH.\arabic{ALG@line}}
\makeatother

\newif\ifnitpick
\nitpickfalse

\newif\ifproofs
\proofstrue

\iflong
\newcommand{\refappendix}[1]{\Cref{#1}}
\else
\newcommand{\refappendix}[1]{the extended version~\cite{extendedVersion}}
\fi

\iflong
\newcommand{\toolong}[1]{#1}
\else
\newcommand{\toolong}[1]{}
\fi

\crefformat{section}{\S#2#1#3} %
\crefformat{subsection}{\S#2#1#3}
\crefformat{subsubsection}{\S#2#1#3}

\Crefname{conjecture}{Conjecture}{Conjectures}
\Crefname{proposition}{Proposition}{Propositions}
\Crefname{lemma}{Lemma}{Lemmas}
\Crefname{corollary}{Corollary}{Corollaries}
\Crefname{example}{Example}{Examples}
\Crefname{definition}{Def.}{Defs.}
\Crefname{algorithm}{Alg.}{Alg.}
\Crefname{theorem}{Thm.}{Thm.}
\Crefname{figure}{Fig.}{Fig.}
\crefname{line}{line}{lines}

\ifproofs
\else
\excludecomment{proof}
\fi

\newtheorem{theo}{Theo}[section] %
\newtheorem{claim}[theo]{Claim}
\newtheorem{remark}[theo]{Remark}
\newcommand{\para}[1]{\vspace{2pt}\noindent\textbf{\textit{#1.}}}

\newcommand{\ov}{\overline}

\newcommand{\A}{\mathcal{A}}

\newcommand{\card}[1]{{\left\vert{#1}\right\vert}} %

\renewcommand{\implies}{\Longrightarrow}

\newcommand{\true}{{\textit{true}}}
\newcommand{\false}{{\textit{false}}}
\newcommand{\vocabulary}{\Sigma}
\newcommand{\voc}{\vocabulary}

\newcommand{\Init}{{\textit{Init}}}
\newcommand{\Bad}{\textit{Bad}}

\newcommand{\TS}{{\textit{TS}}}

\newcommand{\tr}{{\delta}}

\newcommand{\cnf}[2]{\text{\rm CNF}_{#1}^{#2}}
\newcommand{\dnf}[2]{\text{\rm DNF}_{#1}^{#2}}
\newcommand{\moncnf}[2]{\text{\rm Mon-}\text{\rm CNF}_{#1}^{#2}}
\newcommand{\mondnf}[2]{\text{\rm Mon-}\text{\rm DNF}_{#1}^{#2}}

\newcommand{\NP}{\textbf{NP}}

\renewcommand{\vec}{\ov}

\newcommand{\project}[2]{\pi_{#2}({#1})}

\newcommand{\set}[1]{\{{#1}\}}

\newcommand{\cube}[1]{\textit{cube}({#1})}

\newcommand{\prop}{x}

\newcommand{\postimage}[2]{{#1}({#2})}
\newcommand{\bmcunroll}[2]{{#1}^{\leq {#2}}}
\newcommand{\bmc}[3]{\bmcunroll{#1}{#3}({#2})}
\newcommand{\bmcback}[3]{\bmc{({#1}^{-1})}{{#2}}{{#3}}}
\newcommand{\backreachable}[2]{\bmcback{{#1}}{{#2}}{\infty}}

\algdef{SE}[DOWHILE]{Do}{doWhile}{\algorithmicdo}[1]{\algorithmicwhile\ #1}%

\newcommand{\cubemon}[2]{\textit{cube}_{{#2}}({#1})}
\newcommand{\monox}[2]{\mathcal{M}_{#2}({#1})}

\newcommand{\interpolant}{\chi}

\newcommand{\neighborhood}[1]{N({#1})}
\newcommand{\boundarypos}[1]{\partial^{+}({#1})}
\newcommand{\boundaryneg}[1]{\partial^{-}({#1})}
\newcommand{\bigO}{O}
\newcommand{\instr}[1]{\widehat{#1}}

\newcommand{\code}[1]{\ensuremath{\mathtt{#1}}}

\newcommand{\equivalencequery}[1]{{\rm EquivalenceQuery}\left({#1}\right)}
\newcommand{\membershipquery}[1]{{\rm MembershipQuery}\left({#1}\right)}

\newcommand{\ourinterpolationalgname}{\mbox{\rm ITP-Inference-TermMin}}
\newcommand{\dualourinterpolationalgname}{\mbox{\rm Dual-ITP-Inference-ClauseMin}}
\newcommand{\bshoutyinferencealgplain}{$\Lambda$-Inference}
\newcommand{\bshoutyinferencealg}{\mbox{\rm \bshoutyinferencealgplain}}

\begin{document}

\newif\ifcomments
\commentsfalse
\nochangebars
\definecolor{dg}{cmyk}{0.60,0,0.88,0.27}

\newcommand{\sharonnew}[1]{\sharon{#1}}
\newcommand{\yotamnew}[1]{\yotamsmall{#1}}

\ifcomments
\newcommand{\artem}[1]{{\footnotesize\color{olive}[{\bf Artem}: #1]}}
\newcommand{\yotamsmall}[1]{{\footnotesize\color{magenta}[{\bf Yotam}: #1]}}

\newcommand{\sharon}[1]{{\textcolor{purple}{SS: {\em #1}}}}
\newcommand{\adam}[1]{{\textcolor{teal}{AM: {\em #1}}}}
\newcommand{\mooly}[1]{{\textcolor{cyan}{MS: {\em #1}}}}
\newcommand{\noam}[1]{{\textcolor{dg}{NR: {\em #1}}}}
\newcommand{\yotam}[1]{{\textcolor{magenta}{{\bf #1}}}}
\newcommand{\TODO}[1]{{\textcolor{red}{TODO: {\em #1}}}}

\else
\newcommand{\sharon}[1]{}
\newcommand{\adam}[1]{}
\newcommand{\mooly}[1]{}
\newcommand{\neil}[1]{}
\newcommand{\yotam}[1]{}
\newcommand{\TODO}[1]{}
\newcommand{\artem}[1]{}
\newcommand{\yotamsmall}[1]{}

\fi

\newcommand{\commentout}[1]{}
\newcommand{\OMIT}[1]{}  
\title{Learning the Boundary of Inductive Invariants}

\author{Yotam M. Y. Feldman}
\affiliation{
  \institution{Tel Aviv University}
  \country{Israel}
}
\email{yotam.feldman@gmail.com}

\author{Mooly Sagiv}
\affiliation{
  \institution{Tel Aviv University}
  \country{Israel}
}
\email{msagiv@acm.org}

\author{Sharon Shoham}
\affiliation{
  \institution{Tel Aviv University}
  \country{Israel}
}
\email{sharon.shoham@gmail.com}

\author{James R. Wilcox}
\affiliation{
  \institution{Certora}
  \country{USA}
}
\email{james@certora.com}

\begin{abstract}
We study the \emph{complexity of invariant inference} and its connections to \emph{exact concept learning}.
We define a condition on invariants and their geometry, called the \emph{fence} condition, which permits applying theoretical results from exact concept learning to answer open problems in invariant inference theory.
The condition requires the invariant's \emph{boundary}---the states whose Hamming distance from the invariant is one---to be backwards reachable from the bad states in a small number of steps.
Using this condition, we obtain the first polynomial complexity result for an interpolation-based invariant inference algorithm, efficiently inferring monotone DNF invariants with access to a SAT solver as an oracle.
We further harness Bshouty's seminal result in concept learning to efficiently infer invariants of a larger syntactic class of invariants beyond monotone DNF.
Lastly, we consider the robustness of inference under program transformations. We show that some simple transformations preserve the fence condition, and that it is sensitive to more complex transformations.
\end{abstract}  


\maketitle

\section{Introduction}
\label{sec:intro}
This paper addresses the problem of inferring inductive invariants to verify the safety of systems, which lies at the foundation of software and hardware verification. This problem has high complexity, even in finite-state systems~\cite{DBLP:conf/cade/LahiriQ09}\sharon{this sentence (judging by the refs) already considers the SAT-based case. Move to after next sentence?}\yotamsmall{perhaps remove this sentence? I need this not to clash with ``the complexity is not well understood''}\sharon{suggestion (if you don't like it, just keep current version; don't remove the sentence):\\
We are particularly interested in SAT-based inductive invariant inference, which harnesses the SAT solver in order to check if candidate invariant is inductive and to perform bounded model checking. SAT-based invariant inference has high complexity, even in finite-state systems~\cite{DBLP:conf/cade/LahiriQ09,DBLP:journals/pacmpl/FeldmanISS20}
Yet, many techniques tackle this challenge....}.
We are particularly interested in SAT-based inductive invariant inference, which harnesses the SAT solver in order to check if candidate invariant is inductive and to perform bounded model checking.
Many techniques tackle this challenge~\cite[e.g.][]{DBLP:conf/cav/McMillan03,ic3,pdr,DBLP:journals/sttt/SrivastavaGF13,DBLP:series/natosec/AlurBDF0JKMMRSSSSTU15,DBLP:conf/tacas/FedyukovichB18,DBLP:conf/oopsla/DilligDLM13}.
Despite their practical achievements, the theoretical understanding of the complexity of these algorithms is lacking.

This work harnesses the theory of exact concept learning to broaden our understanding of the complexity of invariant inference. We focus on the fundamental setting of propositional systems, which is useful also for infinite-state systems through predicate abstraction~\cite{DBLP:conf/cav/GrafS97,DBLP:conf/popl/FlanaganQ02}.

Our results shed light on the cases in which interpolation-based invariant inference~\cite[e.g.][]{DBLP:conf/cav/McMillan03,DBLP:conf/cav/McMillan06} behaves well.
We analyze an interpolation-based algorithm in which interpolants are computed in a model-based fashion~\cite{DBLP:conf/hvc/ChocklerIM12,DBLP:conf/lpar/BjornerGKL13}, which is inspired by IC3/PDR~\cite{ic3,pdr}, %
and prove an upper bound on the number of iterations until convergence to an invariant.
This upper bound is formulated using a condition, which we call the \emph{fence} condition, on the relationship between reachability in the transition system %
and the geometric notion of the \emph{boundary} of the invariant.
Recall that invariants denote sets of states.
The boundary of the invariant $I$ is the set of states whose Hamming distance from $I$ is one---states in $\neg I$ where flipping just a single bit results in a state that belongs to $I$. (See~\Cref{fig:boundary} for an illustration.)
\begin{figure}[t]
\centering
\begin{minipage}{.5\textwidth}
  \centering
  	 \captionsetup{width=.9\textwidth}
    \includegraphics[width=0.8\textwidth]{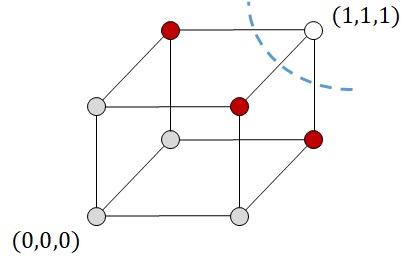}
  	\caption{\footnotesize The (outer) boundary of an invariant $I = x \land y \land z$, denoting the singleton set containing the far-top-right vertex of the 3-dimensional Boolean hypercube, $\set{(1,1,1)}$. Its neighbors are $I$'s boundary (depicted in red): $\set{(1,1,0),(1,0,1),(0,1,1)}$. The rest of the vertices are in $\neg I$ but not in the boundary (depicted in gray). (Illustration inspired by~\cite[][Fig.\ 2.1]{DBLP:books/daglib/0033652}.)
  }
  \label{fig:boundary}
\end{minipage}%
\begin{minipage}{.5\textwidth}
  \centering
  	\captionsetup{width=.9\textwidth}
    \includegraphics[width=0.8\textwidth]{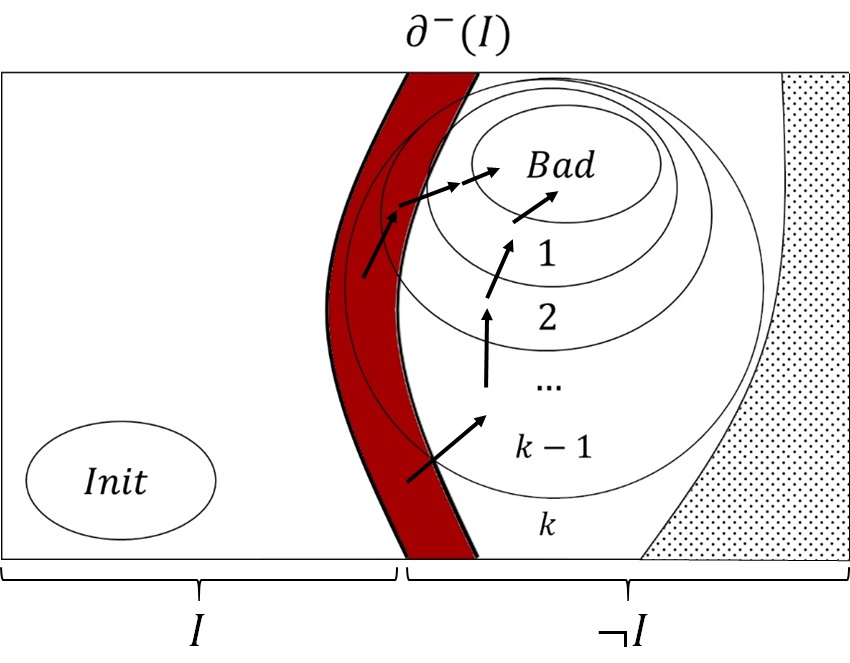}
  	\caption{\footnotesize An illustration of the fence condition. The boundary $\boundaryneg{I}$ of the invariant (the states in $\neg I$ nearest to $I$, in red) are backwards $k$-reachable (reach a bad state in $k$ steps, for example by the transitions depicted by the arrows), but not all states in $\neg I$ are backwards $k$-reachable (or even backwards reachable at all, in the dotted area).
  }
  \label{fig:fence}
\end{minipage}
\vspace{-0.5cm}
\end{figure} The $k$-fence condition requires that all the states in $I$'s boundary reach a bad state in at most $k$ steps. (See~\Cref{fig:fence}.)
Intuitively, the states in the boundary are important, because ``perturbing'' a single bit is the difference between being included in $I$ and being excluded by it (somewhat analogously to error-correcting codes). Through their backwards reachability in $k$ steps,
the fence condition guarantees that the interpolation algorithm, when it uses $k$ for the model checking bound, never ``overshoots'' to invariants that are weaker than $I$ (even though $I$ itself is unknown to the algorithm).
Every safe propositional system admits an invariant that is $k$-fenced for some finite $k$, but this is most useful when (1) $k$ is not prohibitively large, and (2) the invariant has a short representation.

The fence condition opens up the possibility to obtain efficiency results for invariant inference.
We prove that if the fence condition holds and the invariant has a short disjunctive normal form (DNF) representation that is monotone, i.e., all variables appear only positively, then the interpolation algorithm converges in a linear number of iterations.

We are able to transcend this result using a more complex algorithm, in a sense combining multiple instances of the interpolation-based algorithm, based on work in machine learning by~\citet{DBLP:journals/iandc/Bshouty95}. We prove that when the fence condition holds the algorithm
achieves efficient inference for the wider class of \emph{almost-monotone} DNF invariants (DNF invariants that have a constant number of terms that include negative literals).

These theorems
are rooted in results on efficient \emph{exact concept learning} of monotone and almost-monotone DNF formulas by~\citet{DBLP:journals/ml/Angluin87} and~\citet{DBLP:journals/iandc/Bshouty95}, respectively.
In exact concept learning, a learner poses queries to a teacher in order to identify an unknown formula. Unfortunately, invariant inference is harder in general, because the classical queries cannot be implemented~\cite{ICELearning,DBLP:journals/pacmpl/FeldmanISS20}. We show that the fence condition can alleviate these problems when the exact learning algorithm is guaranteed to query only on examples that are positive or have Hamming distance one from a positive example.
In fact, we show a general transformation from algorithms for exact concept learning into invariant inference algorithms that obtain analogous correctness results when the fence condition holds. Such transformations are surprising, because in invariant inference, unlike classical concept learning, the teacher---the SAT solver in our case---does not know the target ``concept'' $I$, and yet it is still able to ``teach'' $I$ to the algorithm, thanks to the fence condition.
A similar transformation, applied to the CDNF algorithm by~\citet{DBLP:journals/iandc/Bshouty95}, provides an algorithm that efficiently infers any invariant that has a short representation in both conjunctive- and disjunctive-normal form (including all small Boolean decision trees) if it satisfies a stronger, two-sided version of the fence condition.
Overall, the fence condition bridges the study of efficient exact concept learning and the complexity analysis of invariant inference.

One of the most significant open problems in formal verification is the \emph{robustness} of the inference method with respect to revisions in the code. For example, suppose the inference algorithm succeeded in inferring an invariant, will this hold after a simple change? This problem is very difficult in general; we address it in the context of the interpolation algorithm, and study whether the fence condition holds after a transformation of the program. We show that very simple transformations preserve the fence condition: variable renaming, and taking conjunctions and disjunctions of properties. In contrast, after a new variable is introduced to capture some meaning over the existing variables (in the hope of aiding inference), an invariant that does not refer to the new variable cannot satisfy the fence condition. Similarly, a natural transformation that attempts to reduce non-monotone invariants to monotone ones is unsuccessful. %

Our results provide an important step towards understanding the complexity of invariant inference algorithms.
It is known that restricting invariants to conjunctions enables efficient inference~\cite{DBLP:conf/fm/FlanaganL01,DBLP:conf/cade/LahiriQ09}.
This work shows that efficient invariant inference can be made possible for invariants of more complex syntactic forms through properties of the transition system and the invariant. %
Our results go beyond the recent upper bound of~\citet{DBLP:journals/pacmpl/FeldmanISS20} for maximal systems in that in our results, the diameter need not be one, not all states in the invariant need to be $k$-reachable, and the invariant can be richer than monotone DNF.

\iflong
\else

\begin{figure*}[t]
  \centering
\begin{minipage}{\textwidth}
  \begin{minipage}{0.55\textwidth}
  \begin{lstlisting}[numbersep=5pt, escapeinside={(*}{*)}, xleftmargin=3.0ex]
init $x_1 = \ldots = x_n = \true$

repeat:
      | hot_potato()
      | turn_two_on()
      | havoc_others()
   
   assert $\neg(x_1 = \false, x_2 = \ldots = x_n = \true)$

havoc_others():
   forall $x_i \not\in J$:
      $x_i$ := *
  \end{lstlisting}
\end{minipage}
  \begin{minipage}{0.35\textwidth}
  \begin{lstlisting}[numbersep=5pt, escapeinside={(*}{*)}, xleftmargin=3.0ex]
hot_potato():
   let $x_i \neq x_j \in J$
   if $x_i=\false \land x_j=\true$:
      $x_i$ := $\true$
      $x_j$ := $\false$

turn_two_on():
   let $x_i \neq x_j \in J$
   if $x_i=\false \land x_j=\false$:
      $x_i$ := $\true$
      $x_j$ := $\true$
\end{lstlisting}
\end{minipage}
\iflong
\else
\vspace{-0.4cm}
\fi
\captionof{figure}{
\footnotesize An example propositional transition system for which we would like to infer an inductive invariant. The state is over $x_1,\ldots,x_n$. $J \subseteq \set{x_1,\ldots,x_n}$, and we assume $x_1 \in J$. In each iteration, one of the actions (after repeat) is invoked, chosen nondeterministically. $\code{forall}$ is executed in a single step of the system. \code{x := *} means that \code{x} is updated to either true or false nondeterministically.}
  \label{fig:parity}
  \end{minipage}
\iflong
\else
\vspace{-1.0cm}
\fi
\end{figure*}  \fi

Overall, this paper makes the following contributions:
\begin{itemize}
	\item We introduce the fence condition, relating a geometric aspect of the invariant to reachability in the system (\Cref{sec:boundary}).

	\item We prove that when the fence condition holds, an interpolation-based algorithm can efficiently infer short monotone DNF invariants (\Cref{sec:itp-all}). A dual algorithm infers short (anti)monotone CNF invariants.
	This is the first result \begin{changebar}showing a polynomial upper bound on the complexity of an interpolation-based invariant inference algorithm with access to a SAT solver as an oracle\end{changebar}.

	\item We show that an algorithm, based on the work by~\citet{DBLP:journals/iandc/Bshouty95}, can efficiently infer short almost-monotone DNF invariants when the fence condition holds (\Cref{sec:beyond-monotone}).
	A dual algorithm infers short almost-(anti)monotone CNF invariants.

	\item We show that these results can be obtained by a translation of algorithms from exact concept learning admitting special conditions, made possible by the fence condition (\Cref{sec:exact-invariant-learning}).

	\item We study the robustness of the fence condition under program transformations (\Cref{sec:robustness}).
\end{itemize}  %
\section{Overview}
\label{sec:overview}

\TODO{we care about worst case}

\iflong

\begin{figure*}[t]
  \centering
\begin{minipage}{\textwidth}
  \begin{minipage}{0.55\textwidth}
  \begin{lstlisting}[numbersep=5pt, escapeinside={(*}{*)}, xleftmargin=3.0ex]
init $x_1 = \ldots = x_n = \true$

repeat:
      | hot_potato()
      | turn_two_on()
      | havoc_others()
   
   assert $\neg(x_1 = \false, x_2 = \ldots = x_n = \true)$

havoc_others():
   forall $x_i \not\in J$:
      $x_i$ := *
  \end{lstlisting}
\end{minipage}
  \begin{minipage}{0.35\textwidth}
  \begin{lstlisting}[numbersep=5pt, escapeinside={(*}{*)}, xleftmargin=3.0ex]
hot_potato():
   let $x_i \neq x_j \in J$
   if $x_i=\false \land x_j=\true$:
      $x_i$ := $\true$
      $x_j$ := $\false$

turn_two_on():
   let $x_i \neq x_j \in J$
   if $x_i=\false \land x_j=\false$:
      $x_i$ := $\true$
      $x_j$ := $\true$
\end{lstlisting}
\end{minipage}
\iflong
\else
\vspace{-0.4cm}
\fi
\captionof{figure}{
\footnotesize An example propositional transition system for which we would like to infer an inductive invariant. The state is over $x_1,\ldots,x_n$. $J \subseteq \set{x_1,\ldots,x_n}$, and we assume $x_1 \in J$. In each iteration, one of the actions (after repeat) is invoked, chosen nondeterministically. $\code{forall}$ is executed in a single step of the system. \code{x := *} means that \code{x} is updated to either true or false nondeterministically.}
  \label{fig:parity}
  \end{minipage}
\iflong
\else
\vspace{-1.0cm}
\fi
\end{figure*}  \fi
Suppose we would like to automatically prove the safety of the transition system in~\Cref{fig:parity}.
In this simple example, there are $n$ propositional variables $x_1,\ldots,x_n$, and $J$ is a subset of the variables so that $x_1 \in J$. The transition system starts with all variables $\true$, and in each iteration either
\begin{inparaenum}
	\item moves a $\false$ value from one $x_i \in J$ to a different $x_j \in J$ that is not already $\false$,
	\item turns two $x_i,x_j \in J$ that are $\false$ to $\true$, or
	\item assigns arbitrary $\true$/$\false$ values to all variables not in $J$.
\end{inparaenum}
The safety property is that the system can never reach the state where $x_1=\false$ and $x_2,\ldots,x_n$ are $\true$.
Given the transition system and safety property, we would like to find an inductive invariant that establishes it.
One inductive invariant here is that the variables in $J$ are always $\true$:
\begin{equation}
\label{eq:overview-parity-example-inv}
I = \bigwedge_{x_i \in J}{x_i}.
\end{equation}
(For other variables, $x_i \not\in J$, this is not true, as $x_i$ can become $\false$ in~\code{havoc\_others}.)

In invariant inference, the goal is to find such inductive invariants completely automatically.
The main motivation of our investigation is the \textbf{\textit{theoretical understanding of when invariant inference algorithms are successful}}.
To this end, we study the complexity of invariant inference algorithms stemming from the seminal \textbf{\textit{interpolation-based approach to inference}} by~\citet{DBLP:conf/cav/McMillan03}. Our formal setting is the problem of finding invariants for propositional transition systems, a fundamental setting which is also relevant for infinite-state systems through predicate abstraction~\cite{DBLP:conf/popl/FlanaganQ02,DBLP:conf/cav/GrafS97}. %
\TODO{say somewhere we don't care about bugs. only weak people have bugs.}

\subsection{Interpolation-Based Invariant Inference}
We start our investigation with the pioneering interpolation-based algorithm~\cite{DBLP:conf/cav/McMillan03}, depicted in~\Cref{alg:interpolation-abstract}.
The algorithm operates by forward-exploration, in each iteration weakening (adding more states) to the current candidate $I$, starting from $I = \Init$.
The algorithm chooses an unrolling bound $k$, and asserts that $I$ cannot reach $\Bad$ in $k$ steps (\cref{ln:itp-abs:check-no-overapprox}). If the check fails, $I$ is too weak---it includes states that cannot be part of any inductive invariant---and the algorithm restarts, with a larger unrolling bound. \begin{changebar}(Unless this happens already when $I = \Init$, indicating that the system is unsafe; we omit this case here for brevity.)\end{changebar}
Otherwise, the algorithm computes an \emph{interpolant} $\interpolant$: a formula ranging over program states that
\begin{inparaenum}
	\item overapproximates the post-image of $I$, namely, includes all states that are reachable in one step from $I$; and
	\item does not include any state that can reach $\Bad$ in $k-1$ steps. %
\end{inparaenum}
In the original work, $\interpolant$ is obtained from a Craig interpolant~\cite{craig1957} of a bounded model checking (BMC) query, which can be computed from the SAT solver's proof~\cite{DBLP:conf/cav/McMillan03}.
The candidate invariant is weakened by taking a disjunction with the interpolant.
In this way, each iteration adds to the invariant at least all the states that are reachable in one step from the current candidate but are not part of it (which are counterexamples to its induction), guided by BMC.
\noindent
\begin{footnotesize}
\begin{minipage}[t]{0.55\textwidth}
\vspace{-0.5cm}
\begin{algorithm}[H]
\caption{\\Interpolation-based invariant inference\\\cite{DBLP:conf/cav/McMillan03}}
\label{alg:interpolation-abstract}
\begin{algorithmic}[1]
\Procedure{Itp-Inference}{$\Init$, $\tr$, $\Bad$, $k$}
	\State $I \gets \Init$
	\While{$I$ not inductive}
		\If{$\bmc{\tr}{I}{k} \cap \Bad \neq \emptyset$} $\label{ln:itp-abs:check-no-overapprox}$
			\State \textbf{restart} with larger $k$
		\EndIf
		\State find $\interpolant$ such that $\postimage{\tr}{I} \implies \interpolant$, $\bmc{\tr}{\interpolant}{k-1} \cap \Bad = \emptyset$
		\State $I \gets I \lor \interpolant$
	\EndWhile
	\State \Return $I$
\EndProcedure
\end{algorithmic}%
\end{algorithm}%
\end{minipage}%
\begin{minipage}[t]{0.45\textwidth}
\vspace{-0.5cm}
\begin{algorithm}[H]
\caption{\\Interpolation by term minimization\\\cite{DBLP:conf/hvc/ChocklerIM12,DBLP:conf/lpar/BjornerGKL13}}
\label{alg:interpolant-from-sample-minimization}
\begin{algorithmic}[1]
\Procedure{TermMinItp}{$I$, $\tr$, $\Bad$, $k-1$}
	\State $\interpolant$ $\gets$ $\false$ %
	\While{$\postimage{\tr}{I} \not\implies \interpolant$}
		\State \textbf{let} $\sigma'$ such that $\sigma' \in \postimage{\tr}{I}$, $\sigma' \not\models \interpolant$
		\If{$\bmc{\tr}{\sigma'}{k-1} \cap \Bad \neq \emptyset$} $\label{ln:abs-itp:check-no-overapprox}$
			\textbf{fail}
		\EndIf
		$d$ $\gets$ $\cube{\sigma'}$
		\For{$\ell$ in $d$}
			\If{$\bmc{\tr}{d \setminus \set{\ell}}{k-1} \cap \Bad = \emptyset$} $\label{ln:abs-itp:bmc-gen}$
				\State $d$ $\gets$ $d \setminus \set{\ell}$
			\EndIf
		\EndFor
		\State $\interpolant \gets \interpolant \lor d$
	\EndWhile
	\State \Return $\interpolant$
\EndProcedure
\end{algorithmic}
\end{algorithm}
\end{minipage}%
\end{footnotesize}

As the original paper shows, with a sufficiently large $k$, the algorithm is guaranteed to converge to an inductive invariant. However, this may take exponentially many iterations~\cite{DBLP:conf/cav/McMillan03}.
From a complexity-theoretic perspective, this guarantee can also be achieved by a simple naive search.
Clearly, the answer as to why interpolation-based inference is better lies with the virtue of interpolants.
However, each bounded proof may allow several interpolants, all satisfying the requirements from $\interpolant$ above, which nonetheless greatly differ from the perspective of invariant inference. Some may be desired parts of the inductive invariant, others might include also states that can reach $\Bad$ and should not be used. Still other interpolants are safe, but using them would lead to very slow convergence to an inductive invariant.
Choosing ``good'' interpolants is the problem of \emph{generalization}: how should the algorithm choose interpolants so that it finds an invariant quickly?
The present view is that generalization strives to abstract away from irrelevant aspects, which is ``heuristic in nature''~\cite{DBLP:reference/mc/McMillan18}. Perhaps for this reason, there is currently no theoretical understanding of the efficiency of interpolation-based algorithms.
In contrast, we identify certain conditions that facilitate a theoretical complexity analysis of this algorithmic approach.

\label{sec:overview-itp-original-analysis}
The challenges in a theoretical complexity analysis of this approach are best understood by considering the (termination) analysis of~\Cref{alg:interpolation-abstract} by~\citet{DBLP:conf/cav/McMillan03}, which is essentially as follows.
Suppose $k$ was increased sufficiently to match the \emph{co-reachability-diameter} of the system: the maximal number of steps required for a state to reach a bad state, if it can reach any bad state.\footnote{\citet{DBLP:conf/cav/McMillan03} calls this number the ``reverse-depth''.}
Then \emph{any} choice of interpolant $\interpolant$ as above is a safe overapproximation in the sense that no state in $\interpolant$ can reach $\Bad$ in any number of steps (even greater than $k$). The candidate $I$ thus never includes states that can reach $\Bad$; put differently, it is below the greatest fixed-point (gfp)---the weakest invariant, consisting of all states that cannot reach $\Bad$. \begin{changebar}Because $I$ is always strictly increasing (becoming strictly weaker) and there is only a finite (albeit exponential) number of strictly increasing formulas\end{changebar}, \Cref{alg:interpolation-abstract} must converge, and at that point $I$ is an inductive invariant.

From the perspective of \emph{complexity}, this termination analysis has several shortcomings:\sharon{maybe instead of "shortcomings", "unanswered questions"?}: %

\begin{enumerate}
	\item \label{it:overview-question-1} \textbf{Short invariants}:
	Are there cases where the algorithm is guaranteed to find invariants that are not the gfp?
	The gfp captures the set of backwards-reachable state in an exact way, but often this is too costly, and unnecessary.
	For example, in the system of~\Cref{fig:parity}, the states that can reach $\Bad$ are the states where there is an odd number of $\false$ variables in $J$. The gfp is thus the invariant saying that the number of $\false$ variables in $J$ is even, whose minimal DNF representation is exponentially long~\cite[][Theorem 3.19]{DBLP:books/daglib/0028067}.
	Indeed, invariant inference typically strives to achieve an invariant which is ``just right'' for establishing the safety property of interest, but this is not reflected by the existing theoretical analysis. Can the algorithm benefit from the existence of an invariant that has a short representation (whether or not it is the gfp)? %

	\item \label{it:overview-question-2} \textbf{Number of iterations}: If the invariant the algorithm finds has a short representation, does the algorithm find it in a small number of iterations? As an illustration, the algorithm could, in principle, resort to inefficiently enumerating the states in the invariant one by one, even though there are more compact ways to represent the invariant. %

	\item \label{it:overview-question-3} \textbf{Unrolling depth}: Is it actually necessary to use the BMC bound $k$ that is as large as the co-diameter?
	For example, in the system of~\Cref{fig:parity}, the co-diameter is $\Theta\left(n\right)$.\footnote{The state with a maximal odd number of variables in $J$ except $x_1$ having value $\false$ requires $\left\lfloor \frac{\card{J}-2}{2} \right\rfloor$ iterations of $\code{turn\_two\_off}$ to turn two such variables to $\true$ until only one is left $\false$, an iteration of $\code{hot\_potato}$ to move it to $x_1$, and another to turn the variables not in $J$ to $\true$.}
	This BMC bound could be prohibitively large for the SAT solver. Worse still, the co-diameter can be exponential in general~\cite{DBLP:conf/cav/McMillan03}.
	Can the algorithm always succeed even when using smaller $k$?
\end{enumerate}

In this work we develop a \textbf{\textit{theoretical analysis of an interpolation-based algorithm}} that addresses all these points.
The analysis of this algorithm characterizes the invariants that it finds, the number of iterations until convergence, and how large the unrolling bound must be. %
In this analysis, the inference problem is ``well-behaved'' if there exists a \emph{short invariant} of \emph{certain syntactic forms} whose \emph{Hamming-geometric boundary} is backwards $k$-reachable \begin{changebar}from $\Bad$\end{changebar}, where $k$ is
a bound which can be \emph{smaller than the co-diameter}. If this is the case, we show algorithms that construct and maintain generalizations in specific ways---connected to \emph{exact learning theory}---which are guaranteed to \begin{changebar}\emph{find inductive invariants efficiently} with a SAT oracle\end{changebar}.

\subsubsection{Interpolation by covering and term minimization}
The algorithmic scheme of~\Cref{alg:interpolation-abstract} itself is not amenable to such an analysis, because
valid choices for $\interpolant$ range between the exact post-image $\postimage{\tr}{I}$ and the set of states that cannot reach $\Bad$ in $k$ steps, potentially leading to very different outcomes over the algorithm's run.
We thus examine a specific method of interpolant construction displayed in~\Cref{alg:interpolant-from-sample-minimization}, due to~\citet{DBLP:conf/hvc/ChocklerIM12,DBLP:conf/lpar/BjornerGKL13} and inspired by IC3/PDR~\cite{ic3,pdr}. The procedure iteratively samples states from the post-image of $I$ that should be added to the interpolant $\interpolant$. Adding a single state exactly, by disjoining the cube of the state $\cube{\sigma'}$---the conjunction of all literals that hold in the state---would converge slowly. Instead, the procedure drops literals from $\cube{\sigma'}$ (thereby including more states) as long as no state that satisfies the remaining conjunction can reach $\Bad$ in $k-1$ steps,\footnote{
	In practice, this can be made more efficient by first obtaining an \emph{unsat core} of the unsat BMC query, and then explicitly dropping literals one by one. Empirically, \citet{DBLP:conf/hvc/ChocklerIM12} found that the extra time spent in explicit minimization after the unsat core is compensated by fewer iterations. %
}
and then disjoins the result to the interpolant. %
By construction, all the states in $\interpolant$ cannot reach $\Bad$ in $k-1$ steps, and the procedure terminates when all of $\postimage{\tr}{I}$ is covered by $\interpolant$. It may seem inefficient to overapproximate from a single state in each iteration, compared to proof-based methods; %
however, our results show that under certain conditions each such iteration makes significant progress.

\subsection{The Complexity of Interpolation-Based Inference}
The interpolation-based inference of~\Cref{alg:interpolation-abstract} with the interpolation procedure of~\Cref{alg:interpolant-from-sample-minimization} produce invariants in disjunctive normal form (DNF). (For example, the invariant in~\Cref{eq:overview-parity-example-inv} is a DNF formula with one term.)
If the shortest invariant in DNF is exponentially long, clearly an exponential number of iterations is necessary.
{\textbf{\textit{Suppose that the system admits a short DNF invariant $I$. When is the interpolation-based inference algorithm guaranteed to be efficient?}}
By efficient we mean that the number of \begin{changebar}steps the algorithms performs is polynomial in the number of variables and the length of $I$, where each $k$-BMC query is counted as a single step of calling a SAT oracle\end{changebar}.
This is a difficult question, because even if an invariant with a short representation exists, the algorithm might miss it or learn a longer representation. %
Our solution is based on two ingredients:
\begin{inparaenum}
	\item the \emph{boundary} of the invariant, which a condition we call the \emph{fence} condition ties to reachability, and
	\item utilizing the syntactic shape of the invariant, an aspect in which, as we show, ideas from \emph{exact learning} are extremely relevant. %
\end{inparaenum}

\subsubsection{The Boundary of Inductive Invariants}
Our algorithm makes decisions using BMC: whether or not a literal is to be dropped, whether or not some states are to be added to the invariant, is a choice made on the basis of bounded reachability information (because unbounded reachability is unknown). This is why the algorithm might fail and restart with a larger $k$ (\cref{ln:abs-itp:check-no-overapprox}).
The fence condition's guarantees that this would not happen.
The idea is to require that BMC finds useful information when it is invoked on sets that have a special role: the states on the \textbf{\textit{{boundary} of the inductive invariant $I$}}.
The (outer) boundary $\boundaryneg{I}$ is the set of states $\sigma$ that do \emph{not} belong to $I$, but are ``almost'' in $I$: flipping just one bit in $\sigma$ yields a state that \emph{does} belong to $I$. Put differently, the boundary is the set of states in $\neg I$ that have a Hamming neighbor (Hamming distance $1$) in $I$.
For example, the boundary of the invariant in~\Cref{eq:overview-parity-example-inv} consists of the states where $\card{\set{i \in J \, | \, x_i=\false}}=1$: such states are not in $I$, but flipping the single $\false$ variable in $J$ results in a state that \emph{does} belong to $I$.
(See~\Cref{fig:boundary} for an illustration of the boundary of this invariant with $n=3$ and $J = \set{x_1,x_2,x_3}$.)
Note that not all states in $\neg I$ are on the boundary: states with more than one $\false$ variable in $J$ also belong to $\neg I$, but their Hamming distance from $I$ is larger than 1.

\label{sec:overview-example-fence-holds}
The \textbf{\textit{fence condition}} requires that the \textbf{\textit{states in the boundary $\boundaryneg{I}$ reach $\Bad$ in at most $k$ steps}}.
For example, in the system of~\Cref{fig:parity}, the invariant in~\Cref{eq:overview-parity-example-inv} satisfies the fence condition with $k=2$; in at most two steps every state in its boundary reaches $\Bad$ (one step to move the $\false$ value to $x_1$, and another to turn off variables not in $J$).
This property is key for the interpolation algorithm to successfully and consistently find an invariant.
We formally define the boundary and the fence condition in~\Cref{sec:boundary}.
(See~\Cref{fig:fence} for an illustration.)

It should be noted that the fence condition is a property of a specific invariant. %
Some invariants may not satisfy the fence condition for a given $k$ or even for any $k$, depending on whether and how quickly states outside the invariant can reach $\Bad$.
Intuitively, this reflects natural differences between invariants: some invariants are well suited to be discovered by the algorithm, while others are such that the algorithm might overshoot and miss them in some of its executions.
\yotam{discussion of weakening/strengthening invariants in comment}

Every safe system admits an invariant that is $k$-fenced, because the gfp if $k$-fenced for $k$ that is the co-diameter (\Cref{lem:boundary-exists-inv-k}).
The power of the fence condition is in that it allows to prove convergence even with $k$ that is \begin{changebar}strictly\end{changebar} smaller than the co-diameter, \begin{changebar}and\end{changebar} based on invariants that need not be the gfp, which is an important ingredient of our \begin{changebar}efficiency\end{changebar} results.\begin{changebar}\footnote{
	Even for the co-diameter, the existing analysis by~\citet{DBLP:conf/cav/McMillan03} does not derive complexity bounds, while we are able to do so using techniques from exact concept learning.
}\end{changebar}

\subsubsection{Convergence With the Fence Condition and a Sufficient BMC Bound}
Intuitively, when the fence condition does not hold, if \Cref{alg:interpolant-from-sample-minimization} discovers a state $\sigma^{+} \models I$ and checks whether it can drop a literal on which its Hamming neighbor $\sigma^{-} \models \neg I$ disagrees, if $\sigma^{-}$ does not reach $\Bad$ in $k$ steps we could include $\sigma^{-}$ in the candidate even though it is not in $I$.
In~\Cref{sec:itp-fence-basic}, we prove that the fence condition ensures that the candidate invariant our algorithm constructs is always below $I$, and therefore converges to $I$ itself or a stronger invariant, \emph{whichever} counterexamples to induction the solver returns, and \emph{regardless} of the order in which \Cref{alg:interpolant-from-sample-minimization} attempts to drop literals (\Cref{lem:itp-underapprox})---although, without further assumptions which we explore next, this convergence may happen only after many iterations of the algorithm.
But before we elaborate on this, two points merit emphasis:

First, not all the states in $\neg I$ need to reach $\Bad$ in $k$ steps, and $k$ can be smaller than the co-diameter (which addresses question~(\ref{it:overview-question-3}) above).
For example, in~\Cref{fig:parity}, our results show that $k=2$ suffices, whereas the co-diameter is linear in the number of variables.
Further, this deviates from the original termination argument above, and facilitates an analysis where the invariant to which the algorithm converges is not the gfp.
(The latter point already addresses question~(\ref{it:overview-question-1}) above.)

Second, the algorithm does not posses any a-priori knowledge of the invariant $I$; the existence of an invariant as mandated by the fence condition suffices to guarantee convergence.

We are now ready to tackle the question of \emph{efficient convergence} when $I$ is not only $k$-fenced but also belongs to syntactic class of ``manageable'' formulas.

\subsubsection{Efficient Inference of Short Monotone Invariants}
The fence condition guarantees that the generated interpolants underapproximate an invariant $I$. How many interpolants must the algorithm find before it converges to an invariant?
We prove in~\Cref{sec:itp-monotone} that if $I$ is in monotone DNF, that is, all variables appear positively, then the number of iterations of~\Cref{alg:interpolation-abstract} is bounded by the number of terms in $I$ if $I$ is $k$-fenced (\Cref{thm:monotone-inference-efficient})---thus addressing question~(\ref{it:overview-question-2}) above.
\textbf{\textit{Overall, this is a theorem of efficient invariant inference by~\Cref{alg:interpolation-abstract}.}}
For example, the system in~\Cref{fig:parity} admits the short monotone DNF invariant in~\Cref{eq:overview-parity-example-inv}, hence by our results the algorithm efficiently infers an invariant for this system.

\TODO{also add to intro?}
An intriguing consequence of formal efficiency results for an algorithm is that when the algorithm fails to converge, this is a \textbf{\textit{witness that an invariant of a certain type does not exist}}. Thus, if the algorithm continues to execute beyond the number of steps mandated by our upper bound, this means that there is no monotone DNF $k$-fenced invariant with a specified number of terms. This may indicate a bug rendering the system unsafe, or perhaps that an invariant exists but it is not $k$-fenced, not monotone, or too long.

We also derive a dual efficiency result: a \emph{dual-interpolation} algorithm achieves efficient inference for short (anti)monotone CNF (\Cref{thm:monotone-inference-efficient-cnf})---which are CNF invariants where all variables appear negatively---when a \emph{dual fence condition} holds: the inner boundary $\boundarypos{I}$ is the set of states in $I$ that have a Hamming neighbor in $\neg I$, and it must be $k$-reachable from $\Init$.

\subsection{Efficient Inference Beyond Monotone Invariants and Exact Learning Theory}
The interpolation-based algorithm is not guaranteed to perform a small number of iterations when $I$ is not monotone.
\textbf{\textit{Is provably efficient inference beyond monotone invariants possible?}}
In~\Cref{sec:beyond-monotone},
we obtain an efficient inference algorithm for the wider class of \emph{almost-monotone} DNF invariants, which are DNF formulas that have at most $\bigO(1)$ terms that include negative literals, provided that the fence condition is satisfied (\Cref{cor:almost-monotone-dnf}).
This upper bound is achieved by a different, yet related algorithm, which is based on the celebrated work of~\citet{DBLP:journals/iandc/Bshouty95} in machine learning.
Roughly, the algorithm uses several instances of~\Cref{alg:interpolation-abstract} that learn several overapproximations of $I$, each of them monotone under a different translation of the variables. A dual result holds for almost-(anti)monotone CNF invariants and the dual fence condition (\Cref{cor:almost-monotone-cnf}).

Underlying this development is the realization that efficient inference based on the fence condition has close connections to efficient \emph{exact concept learning}~\cite{DBLP:journals/ml/Angluin87}.
In exact concept learning, the goal is to learn an unknown formula $\varphi$ through a sequence of queries to a teacher. The most prominent types of queries are (1) equivalence, in which the learner suggests a candidate $\theta$, and the teacher returns whether $\theta \equiv \varphi$, or a differentiating counterexample; and (2) membership, in which the learner chooses a valuation $v$ and the teacher responds whether $v \models \varphi$.
These queries are hard to implement in an invariant inference setting~\cite{ICELearning,DBLP:journals/pacmpl/FeldmanISS20}.
Nevertheless, we show that the algorithm which achieves~\Cref{cor:almost-monotone-dnf} can be obtained directly from the exact concept learning algorithm by~\citet{DBLP:journals/iandc/Bshouty95} through a transformation which, \emph{when the fence condition holds}, can implement \emph{certain} equivalence and membership queries using BMC (\Cref{sec:restricted-queries}). This is surprising because in invariant inference, unlike classical concept learning, the ``teacher'' (SAT solver) does not know the target ``concept'' $I$. (We also provide a self-contained analysis of this algorithm, based on the monotone theory~\cite{DBLP:journals/iandc/Bshouty95}, in~\Cref{sec:bshouty-direct}.)
It is interesting to note that applying this transformation to the exact concept learning algorithm for monotone DNF formulas by~\citet{DBLP:journals/ml/Angluin87} produces an algorithm that essentially matches~\Cref{alg:interpolation-abstract}+\Cref{alg:interpolant-from-sample-minimization}.

Lastly, we extend our translation from exact learning algorithms to show that \emph{every} exact learning algorithm from equivalence and membership queries can be efficiently implemented, provided that the fence condition and its dual both hold together. This yields an efficient inference algorithm for all invariants that have both a short CNF and a short DNF representation (not necessarily monotonic), including formulas that can be expressed by small Boolean decision trees (\Cref{thm:bshouty-inference-cdnf}), when this two-sided fence condition holds.

\subsection{Robustness and Non-Robustness}
\label{sec:overview:robustness}
One of the most interesting questions for invariant inference as a practical methodology is robustness:
\textbf{\textit{How do modifications to the program affect invariant inference?}}
On the one hand, it is desirable that if an invariant could be found before, and the program undergoes an ``inconsequential'' transformation, then the algorithm would successfully find the invariant also after the transformation.
On the other, when the algorithm does not manage to find an invariant, altering the program can prove a successful strategy to achieve convergence~\cite{TOPLAS:SRW02,DBLP:conf/cav/FeldmanWSS19,DBLP:conf/cav/SharmaDDA11,DBLP:conf/tacas/BorrallerasBLOR17,DBLP:conf/vmcai/Namjoshi07,DBLP:journals/pacmpl/CyphertBKR19}. \TODO{other works?}
The effect of program transformations on inference depends on the inference algorithm.\footnote{
For example, the inference of conjunctive invariants via Houdini~\cite{DBLP:conf/fm/FlanaganL01} is robust, and always finds an invariant if one exists in a polynomial number of SAT calls.
}
In~\Cref{sec:robustness} we study this question from the perspective of the fence condition: if an invariant is $k$-fenced before the transformation, does this still hold after the transformation %
(thereby making our efficiency results applicable)?
We show simple transformations that are robust, ensuring that the invariant is $k$-fenced also after the transformation: variable renaming and translation, and strengthening safety properties.
We then show that interesting transformations that add new variables are not robust: instrumentation with a derived relation, and ``monotonization'' of an invariant using new variables that track negations.

\TODO{what do you think about this example?}
For example, suppose we add a bit $q$ to the example of~\Cref{fig:parity} to represent the parity of the variables in $J$: the xor $\oplus_{x_i \in J}{x_i}$. We initialize $q$ to the correct initial value (which is $\card{J} \pmod{2}$.
The motivation is to use $q$ to prove that the system avoids the bad state, so we consider an error when $x_1=\false,x_2=\ldots=x_n=\true$, \emph{and} $q$ correctly matches the parity of the variables in $J$ in the error state (i.e.\ $q = \card{J}-1 \pmod{2}$). The parity is not changed by any of the actions, so $q$ is not modified in any action. Under this transformation, \Cref{eq:overview-parity-example-inv} is still an inductive invariant, but it is no longer backwards $k$-fenced (for any $k$): a state with exactly one variable $\true$ out of $J$ but with the incorrect parity in $q$ cannot reach the bad state. This suggests that the algorithm cannot be guaranteed to converge to the original invariant. (Indeed, the motivation for the transformation is the different invariant $q=\card{J} \pmod{2}$; alas, in this case, this invariant is also not backwards $k$-fenced: its boundary consists of all states with $q \neq \card{J} \pmod{2}$, but this includes states with an even number of $\true$ variables from $J$, which cannot reach the bad state.)

\begin{changebar}
This non-robustness result matches the way invariant inference behaves in practice and the butterfly effect of introducing a derived relation, which indicates that our theoretical analysis can reproduce some realistic phenomena.
The non-robustness result means that the algorithm is not guaranteed to converge to an invariant that does not use the new variable; nonetheless, invariants that use the new variable may or may not satisfy the fence condition, and this depends on the example. Also in practice, the inference algorithm learns properties that use the new variable, and the transformation may help inference to converge to a new invariant, but might not.
\end{changebar}

\section{Preliminaries}
\label{sec:prelim}

\para{States, transition systems, inductive invariants}
We consider the safety of transition systems defined over a propositional vocabulary $\voc = \set{p_1,\ldots,p_n}$ of $n$ Boolean variables.
A \emph{state} is a \emph{valuation} to $\voc$.
If $x$ is a state, $x[p]$ is the value ($\true$ or $\false$)\sharon{how about: ($\true$/$\false$ or $1/0$)} that $x$ assigns to the variable $p$. We write $x[p \mapsto z]$ for the state obtained from $x$ by assigning the value $z$ to $p$.
A \emph{transition system} is a triple $(\Init,\tr,\Bad)$ where $\Init,\Bad$ are formulas over $\voc$ denoting the set of initial and bad states (respectively), and the \emph{transition relation} $\tr$ is a formula over $\voc \uplus \voc'$, where $\voc' = \{ \prop' \mid \prop \in \voc\}$ is a copy of the vocabulary used to describe the post-state of a transition.
We assume that $\Init \not\equiv \false$, $\Bad \not\equiv \true,\false$, $\Init \implies \neg\Bad$.
A transition system is \emph{safe} if all the states that are reachable from the initial states via any number of steps of $\tr$ satisfy $\neg \Bad$. %
An \emph{inductive invariant} is a formula $I$ over $\voc$ such that
\begin{inparaenum}[(i)]
	\item $\Init \implies I$,
	\item $I \land \tr \implies I'$, and
	\item $I \implies \neg\Bad$, where $I'$ %
denotes the result of substituting each $\prop \in \voc$ for $\prop' \in \voc'$ in $I$,
\end{inparaenum}
and $\varphi \implies \psi$ denotes the validity of the formula $\varphi \to \psi$. In the context of propositional logic, a transition system is safe if and only if it has an inductive invariant. %
An \emph{inductiveness check}, typically implemented by a SAT query, returns whether a given candidate $\varphi$ is an inductive invariant, and a counterexample if it violates one of the implications (i)--(iii) above. When implication (ii) is violated, a \emph{counterexample to induction} is a pair of states $\sigma, \sigma'$ such that $\sigma,\sigma'\models I \land \tr \land \neg I'$ (where the valuation to $\voc'$ is taken from $\sigma'$).

\para{Bounded reachability and bounded model checking}
Given a transition system $(\Init,\tr,\Bad)$ and a \emph{bound} $k \in \mathbb{N}$, the $k$-forward reachable states $\bmc{\tr}{\Init}{k}$ are the states $\sigma$ such that there is an execution of $\tr$ of length at most $k$ starting from a state in $\Init$ and ending at $\sigma$. %
The $k$-backwards reachable states $\bmcback{\tr}{\Bad}{k}$ are the states $\sigma$ such that there is an execution of $\tr$ of length at most $k$ starting from $\sigma$ and ending at a state in $\Bad$.
\emph{Bounded model checking (BMC)}~\cite{DBLP:conf/tacas/BiereCCZ99} checks whether a set of states described by a formula $\psi$ is
\begin{itemize}
	\item forwards unreachable, checking $\bmc{\tr}{\Init}{k} \cap \psi \overset{?}{=} \emptyset$.
	As a SAT query, this is implemented by checking whether
	$\Init(\voc_0) \land \bigwedge_{i=0}^{k-1}{\tr(\voc_i,\voc_{i+1})} \land \left(\bigvee_{i=0}^{k} \psi(\voc_i)\right)$ is unsatisfiable;
	or
	\item backwards unreachable, checking $\bmc{\tr}{\psi}{k} \cap \Bad \overset{?}{=} \emptyset$.
	As a SAT query, this is implemented by checking whether
	$\psi(\voc_0) \land \bigwedge_{i=0}^{k-1}{\tr(\voc_i,\voc_{i+1})} \land \left(\bigvee_{i=0}^{k} \Bad(\voc_i)\right)$ is unsatisfiable.
\end{itemize}
($\voc_0,\ldots,\voc_k$ are distinct copies of the vocabulary $\voc$.)
A $k$-BMC check is either one, with bound $k$.

\para{Forwards/backwards duality}
\label{sec:forward-backward-duality}
A known duality~\cite[e.g][Appendix A]{DBLP:journals/pacmpl/FeldmanISS20} substitutes forwards and backwards reachability:
The dual of a transition system $(\Init,\tr,\Bad)$ is $(\Bad,\tr^{-1},\Init)$, where $\tr^{-1} = \tr(\voc',\voc)$ is the inverse of $\tr$. The dual transition system is safe iff the original one is, and $I$ is an inductive invariant for the original system iff $\neg I$ is an inductive invariant for the dual.
Note that the set of $k$-forwards reachable states in the original system is exactly the set of $k$-backwards reachable states in the dual system.
Given an inference algorithm $\A(\Init,\tr,\Bad)$, the \emph{dual algorithm} is $\A^*(\Init,\tr,\Bad) = \neg \A(\Bad,\tr^{-1},\Init)$.
\iflong
For example, Dual-Houdini for disjunctive invariants~\cite{DBLP:conf/cade/LahiriQ09} is the dual of Houdini~\cite{DBLP:conf/fm/FlanaganL01} which infers conjunctive invariants\sharon{omit Houdini example?}.
\fi
Complexity results translate between the algorithm and its dual:
$\A^*(\Init,\tr,\Bad)$ finds an invariant in $t$ SAT calls if $\A(\Bad,\tr^{-1},\Init)$ finds an invariant in $t$ SAT calls.

\para{CNF, DNF, and cubes}
A \emph{literal} $\ell$ is a variable $p$ (positive literal) or its negation $\neg p$ (negative literal).
A \emph{term} is a conjunction of literals; at times, we also refer to it as a set of literals. A formula is in \emph{disjunctive normal form (DNF)} if it is a disjunction of terms.
A clause is a disjunction of literals. A formula is in \emph{conjunctive normal norm (CNF)} if it is a conjunction of clauses.
The \emph{cube} of a state $\sigma$, denoted $\cube{\sigma}$, is the term that is the conjunction of all the literals that are satisfied in $\sigma$.

\para{Complexity of SAT-based invariant inference}
We measure the complexity of a SAT-based inference algorithm by
\begin{inparaenum}[(i)]
	\item the number of inductiveness checks,
	\item the number of $k$-BMC checks, and
	\item the number of \begin{changebar}other\end{changebar} steps, when each SAT call is considered one step (an oracle call).
\end{inparaenum}
The complexity parameters are the number of variables $n = \card{\voc}$ and syntactic measures of the length of the target invariant.
The inference of anti/monotone CNF/DNF invariants with $n$ clauses/terms is in general $\NP$-hard with access to a SAT solver~\cite{DBLP:journals/pacmpl/FeldmanISS20,DBLP:conf/cade/LahiriQ09}. %
\section{The Boundary of Inductive Invariants}
\label{sec:boundary}

In this section we present the backwards $k$-fenced condition, which is the foundation of our analysis of inference algorithms and all our convergence results.

\begin{definition}[Neighborhood, Boundary]
Two states $\sigma_1,\sigma_2$ are \emph{neighbors} if their Hamming distance is one, namely, $\card{\set{p \in \voc \, | \, \sigma_1 \models p, \sigma_2 \not\models p \mbox{ or } \sigma_1 \not\models p, \sigma_2 \models p}} = 1$.
The \emph{neighborhood} of a state $\sigma$, denoted $\neighborhood{\sigma}$, is the set of neighbors of $\sigma$.
The \emph{inner-boundary} of a set of states $S$ is $\boundarypos{S} = \set{\sigma \in S \mid \neighborhood{\sigma} \cap \bar{S} \neq \emptyset}$ (where $\bar{S}$ is the complement of $S$). %
Note that $\boundarypos{S} \subseteq S$, and the inclusion may be strict (when some $\sigma \in S$ has no neighbors outside of $S$).
The \emph{outer-boundary} is $\boundaryneg{S} = \boundarypos{\bar{S}}$. That is, $\boundaryneg{S} = \set{\sigma \not\in S \mid \neighborhood{\sigma} \cap S \neq \emptyset}$.
\end{definition}

\begin{definition}[Backwards $k$-Fenced]
\label{def:fence-backwards}
For a transition system $(\Init,\tr,\Bad)$, an inductive invariant $I$ is \emph{backwards $k$-fenced} for $k \in \mathbb{N}$ if $\boundaryneg{I} \subseteq \bmcback{\tr}{\Bad}{k}$.
\end{definition}
More explicitly, an invariant $I$ is backwards $k$-fenced if every state in $\neg I$ that has a Hamming neighbor in $I$ can reach $\Bad$ in at most $k$ steps.

\begin{example}
\label{ex:even-increment}
Consider a program that manipulates two numbers represented in binary by $\code{x} = x_1,\ldots,x_n$ and $\code{y} = y_1,\ldots,y_n$. Initially, $\code{x}$ is $\textit{odd}$ ($x_n=1$), and $\code{y}$ is $\textit{even}$ ($y_n=0$). In each iteration, $\code{y}$ is incremented by an even number and $\code{x}$ is incremented by $\code{y}$ (all computations are mod $2^n$). The bad states are those where $\code{x}$ is even. An inductive invariant $I$ states that $\textit{odd}(\code{x}) \land \textit{even}(\code{y})$.

Every state in $\neg I$---and in particular, in $\boundaryneg{I}$---reaches a bad state in at most one step. This is because every state where $\textit{even}(\code{x})$ holds is bad, and every state where $\textit{odd}(\code{y})$ holds is either bad, or the step that adds $\code{y}$ to $\code{x}$ leads to a bad state. Hence, $I$ is $1$-backwards fenced.

Now consider the same system %
except there is a flag $\code{z}$ that decides whether $\code{x}$ is modified; the system takes a step only if $\code{z}=\false$.
The same invariant from before applies, but it is no longer backwards $1$-fenced. This is because the state in $\boundaryneg{I}$ where $\code{x}$ is odd, $\code{y}$ is odd, and $\code{z}=\false$ cannot reach a bad state (nor perform any transition).
However, a different invariant, $\textit{odd}(\code{x}) \land (\textit{even}(\code{y}) \lor \neg\code{z})$ is $1$-backwards fenced in this system.
\end{example}

In every system, this condition holds for at least one inductive invariant and for some finite $k$: the gfp---the weakest invariant, that allows all states but those that can reach $\Bad$ in any number of steps---satisfies the condition with the co-diameter, the number of steps that takes for all states that can reach $\Bad$ to do so.
\begin{lemma}
\label{lem:boundary-exists-inv-k}
Every safe transition system $\TS =(\Init,\tr,\Bad)$ admits an inductive invariant $\textit{gfp} = \neg \left(\backreachable{\tr}{\Bad}\right)$ that is backwards $k$-fenced for $k$ that is the co-diameter: the minimal $k$ such that $\bmcback{\tr}{\Bad}{k} = \backreachable{\tr}{\Bad}$.\iflong\else\footnote{All omitted proofs appear in the supplementary materials.}\fi
\end{lemma}
\toolong{
\begin{proof}
First, note that such a $k$ exists because for every $\sigma \in \backreachable{\tr}{\Bad}$ there is a finite $r$ such that $\sigma \in \bmcback{\tr}{\Bad}{r}$, and $k$ is the maximum over these $r$'s, and there are finitely many of these because the number of states is finite.)
By the choice of $I$ it holds that $\neg I \subseteq \backreachable{\tr}{\Bad}$.
Seeing that $\boundaryneg{I} \subseteq \neg I$, we have obtained $\boundaryneg{I} \subseteq \bmcback{\tr}{\Bad}{k}$, as desired.
\yotam{prelim: $\backreachable{\tr}{\Bad}$}
\end{proof}
}
While this lemma shows the existence of a backwards fenced invariant through the gfp and co-diameter, the $k$-fence condition is more liberal: it can hold also for an invariant when not \emph{every} state in $\neg I$ reaches $\Bad$ in $k$ steps, and only the states in $\boundaryneg{I}$ do. An example demonstrating this appears in~\Cref{sec:overview}. An additional example follows.
\begin{example}
Consider an example of a (doubly)-linked list traversal, using $\code{i}$ to traverse the list backwards, modeled via predicate abstraction following~\citet{DBLP:conf/cav/ItzhakyBRST14}. %
The list starts at $\code{h}$.
Initially, $\code{i}$ points to some location that may or may not be part of the list, and in each step the system goes from $\code{i}$ to its predecessor, until that would reach $\code{x}$.
We write $\code{s} \leadsto r$ to denote that $r$ is reachable from $\code{s}$ by following zero or more links.
Consider the initial assumption $\code{h} \leadsto \code{x}$, but $\code{i} \not\leadsto \code{x}$ (it may be that $\code{x} \leadsto \code{i}$, or that $\code{i}$ is not at all in the list). The bad states are those where $\code{i} = \code{h}$.

\yotam{didn't verify this in mypyvy}
An inductive invariant for this system is $\code{h} \leadsto \code{x} \land \neg \code{i} \leadsto \code{x}$.
In predicate abstraction, we may take the predicates $p_{h,x} = \code{h} \leadsto \code{x}$, $p_{i,x}=\code{i} \leadsto \code{x}$, and write $I = p_{h,x} \land \neg p_{i,x}$, which is a DNF invariant with one term.
Hence $\neg I \equiv \neg p_{h,x} \lor p_{i,x}$.
The outer boundary $\boundaryneg{I}$ consists of the states (1) $p_{h,x}=\false,p_{i,x}=\false$ and (2) $p_{h,x}=\true,p_{i,x}=\true$. Both states are in fact bad states under the abstraction: both include a state where $\code{i}=\code{h}$, from which $\code{x}$ is unreachable (in (1)) or reachable (in (2)). Thus, $I$ is backwards $k$-fenced for every $k \geq 0$.

In contrast, not all the states in $\neg I$ reach bad states (in particular, $I$ is not the gfp): the state $p_{h,x}=\false,p_{i,x}=\true$ abstracts only states where $\code{h} \not\leadsto \code{i}$, and this remains true after going to the predecessor of $\code{i}$. This shows that the fence condition may hold even though $I$ is not the gfp, and not all states in $\neg I$ reach bad states (in $k$ steps or at all).
\end{example}

\section{Efficient Interpolation With the Fence Condition}
\label{sec:itp-all}
In this section we prove that the interpolation-based invariant inference algorithm that computes interpolants using sampling and term minimization is efficient for short monotone DNF invariants.
\Cref{sec:itp-termmin} describes the algorithm. \Cref{sec:itp-fence-basic} derives the algorithm's basic properties and its convergence from the fence condition. \Cref{sec:itp-monotone} builds on this to obtain the efficiency result.

\subsection{Interpolation by Term Minimization}
\label{sec:itp-termmin}
We begin with a formal presentation of the algorithm we will be analyzing in this section, \Cref{alg:itp-termmin}.
\begin{wrapfigure}{L}{0.5\textwidth}
\vspace{-0.72cm}
\begin{minipage}{0.5\textwidth}
\begin{algorithm}[H]
\caption{Interpolation-based inference by term minimization}
\label{alg:itp-termmin}
\begin{algorithmic}[1]
\begin{footnotesize}
\Procedure{\ourinterpolationalgname}{$\Init$, $\tr$, $\Bad$, $k$}
	\State $\varphi \gets \Init$
	\While{$\varphi$ not inductive} $\label{ln:itp-termmin:ind-check}$
		\State \textbf{let} $\sigma, \sigma' \models \varphi \land \tr \land \neg \varphi'$ $\label{ln:itp-termmin:cti}$
		\If{$\bmc{\tr}{\sigma'}{k} \cap \Bad \neq \emptyset$} $\label{ln:itp-termmin:check-no-overapprox}$
			\State \textbf{restart} with larger $k$ $\label{ln:itp-termmin:fail}$
		\EndIf
		$d$ $\gets$ $\cube{\sigma'}$ $\label{ln:itp-termmin:gen-start}$
		\For{$\ell$ in $d$}
			\If{$\bmc{\tr}{d \setminus \set{\ell}}{k} \cap \Bad = \emptyset$} $\label{ln:itp-termmin:bmc-gen}$
				\State $d$ $\gets$ $d \setminus \set{\ell}$
			\EndIf
		\EndFor $\label{ln:itp-termmin:gen-end}$
		\State $\varphi \gets \varphi \lor d$ $\label{ln:itp-termmin:add-disjunct}$
	\EndWhile
	\State \Return $I$
\EndProcedure
\end{footnotesize}
\end{algorithmic}
\end{algorithm} 
\end{minipage}
\vspace{-0.4cm}
\end{wrapfigure}
It is a simplification of the one presented in~\Cref{sec:overview}, ``merging'' the two loops formed when \Cref{alg:interpolation-abstract} uses \Cref{alg:interpolant-from-sample-minimization}; instead of first computing an overapproximation $\interpolant$ of the post-image of the entire previous candidate %
and only then disjoining $\interpolant$, \Cref{alg:itp-termmin} disjoins the generalization $d$ to the candidate immediately, so the next counterexample to induction may use pre-states from this generalization, rather from the previous candidate. Our results apply equally also to \Cref{alg:interpolation-abstract}+\Cref{alg:interpolant-from-sample-minimization}.

\Cref{alg:itp-termmin} starts with the candidate invariant $\varphi = \Init$, which is gradually increased to include more states.
In each iteration, the algorithm performs an inductiveness check (\cref{ln:itp-termmin:ind-check,ln:itp-termmin:cti}), implemented by SAT calls, and terminates if an inductive invariant has been found.
If a counterexample to induction $(\sigma,\sigma')$ exists, the algorithm generates a term $d$ which includes the post-state $\sigma'$, and disjoins $d$ to $\varphi$ to obtain the new candidate (\cref{ln:itp-termmin:add-disjunct}).
We refer to $d$ as the \emph{generalization} obtained from $\sigma'$.
Starting with $\cube{\sigma'}$---the conjunction that exactly captures $\sigma'$---the algorithm drops literals as long as no state in $d$ can reach a bad state in $k$ steps or less (\cref{ln:itp-termmin:bmc-gen}). These checks invoke the SAT solver with BMC queries.
If $\sigma'$ itself reaches a bad state in $k$ steps, no invariant weaker than $\varphi$ exists, and the algorithm restarts with a larger bound $k$ (\cref{ln:itp-termmin:fail}).
The soundness of this algorithm is immediate: $\varphi$ always includes $\Init$, excludes $\Bad$ (otherwise the algorithm restarts at~\cref{ln:itp-termmin:fail}),
and stops when there is no counterexample to induction.
\subsection{Interpolation Confined in the Boundary}
\label{sec:itp-fence-basic}
We now show how the fence condition ensures that \Cref{alg:itp-termmin} does not ``overshoot'' beyond the inductive invariant when it uses a large enough BMC bound.
We use this to derive a termination property, which we use in~\Cref{sec:itp-monotone} to obtain the efficiency result.

``Not overshooting'' beyond $I$ is formalized in the following lemma:
\begin{lemma}
\label{lem:itp-underapprox}
	Let $(\Init,\tr,\Bad)$ be a transition system, $I$ an inductive invariant, and $k \in \mathbb{N}$.
	If $I$ is backwards $k$-fenced, then throughout the execution of \ourinterpolationalgname($\Init,\tr,\Bad,k$), every candidate $\varphi$ is an underapproximation of $I$, namely, $\varphi \implies I$.
\end{lemma}
\begin{proof}%
\sharon{skipped proof}
By induction on the algorithm's iterations: initially, $\varphi = \Init \implies I$ since $I$ is an inductive invariant; for later iterations, we show that every $d \implies I$. By induction on the iterations in generalization:
when it starts, $d = \sigma'$ satisfies this, because $\varphi \implies I$ so $\sigma \models I$, hence also $\sigma' \models I$ because $I$ is inductive. Later, assume there is a point when $d \implies I$ stops holding: the algorithm drops a literal $\ell$, obtaining $\tilde{d} = d \setminus \set{\ell}$ where $d \implies I$ but $\tilde{d} \not\implies I$, and the check passes: $\bmc{\tr}{\tilde{d}}{k} \cap \Bad = \emptyset$.
Let $\sigma_2 \models \tilde{d} \land \neg I$, and let $\sigma_1$ a state which differs from $\sigma_2$ on the variable in $\ell$ alone.
Necessarily $\sigma_2 \not\models \ell$ (because $\sigma_2 \not\models d$ as $d \implies I$), so $\sigma_1 \models \ell$. The other literals in $\tilde{d}$ are also satisfied by $\sigma_1$ because they are satisfied by $\sigma_2$ and $\sigma_1,\sigma_2$ do not differ there. Thus $\sigma_1 \models d$ and in particular $\sigma_1 \models I$. We have thus obtained Hamming neighbors $\sigma_1,\sigma_2$ such that $\sigma_1 \models I, \sigma_2 \not\models I$, but $\bmc{\tr}{\tilde{d}}{k} \cap \Bad = \emptyset$ is a contradiction to $I$ being backwards $k$-fenced, seeing that $\sigma_2 \models \tilde{d}$.
\end{proof}

This lemma implies that when the condition holds, the invariant $I$ acts as a ``barrier'' from unsafe overgeneralization, and the algorithm does not fail (\cref{ln:itp-termmin:fail} does not execute). This holds even though $k$ may be smaller than the co-diameter (see~\Cref{sec:overview-itp-original-analysis}), as long as there exists an $I$ which is backwards $k$-fenced. ($I$ is not known to the algorithm.)
Thus, in such a case, the algorithm successfully finds an inductive invariant that is an underapproximation of $I$.
(Without further assumptions, this might take exponentially many steps, a challenge which is the focus of~\Cref{sec:itp-monotone}.)
\begin{lemma}
\label{lem:no-fail}
\label{lem:itp-converge-underapprox}
	Let $(\Init,\tr,\Bad)$ be a transition system and $k \in \mathbb{N}$. If there exists an inductive invariant $I$ that is backwards $k$-fenced,
	then \ourinterpolationalgname($\Init,\tr,\Bad,k$) successfully (albeit potentially in an exponential number of steps) finds an inductive invariant $\varphi$ such that $\varphi \implies I$.
\end{lemma}
\toolong{
\begin{proof}
We first claim that no execution %
fails (\cref{ln:itp-termmin:fail} of~\Cref{alg:itp-termmin} does not execute).
As in the proof of~\Cref{lem:itp-underapprox}, $\sigma' \models I$ as $\varphi \implies I$ and $I$ is inductive.
Since $I$ is an inductive invariant, $\sigma'$ cannot reach $\Bad$ in any number of steps, and in particular $\sigma' \not\in \bmcback{\tr}{\Bad}{k}$.

Since $\varphi$ grows monotonically, and there are only finitely many formulas over the fixed vocabulary, $\varphi$ must converge. When this happens $I$ is inductive invariant, and we have $\varphi \implies I$ by~\Cref{lem:itp-underapprox}.
\end{proof}
}

Recalling that the gfp is backwards $k$-fenced for $k$ which is at most the co-diameter (\Cref{lem:boundary-exists-inv-k}), this yields a completeness result, akin to the completeness result by~\citet{DBLP:conf/cav/McMillan03}. %
\begin{corollary}
\label{cor:itp-termmin-complete}
Let $(\Init,\tr,\Bad)$ be a safe transition system. Then there is a bound $k \in \mathbb{N}$ such that \ourinterpolationalgname($\Init,\tr,\Bad,k$) successfully finds an inductive invariant (albeit potentially in an exponential number of steps).
\end{corollary}

So far we have provided an upper bound on the $k$ needed for convergence which may be smaller than the co-diameter.
(Even the gfp can be $k$-fenced with $k$ that is smaller than the co-diameter!)
The real power of the approach, however, lies in the complexity analysis the fence condition facilitates beyond this completeness result, which we carry out in the next sections. %

\begin{remark}
\label{rem:fence-violation}
What happens to~\Cref{alg:itp-termmin} when the fence condition does not hold? Suppose there are two states $\sigma^{-} \models \neg I, \sigma^{+} \models I$ that differ in a single variable $p_i$. The backwards fence condition requires that $\sigma^{-}$ reaches $\Bad$ in $k$ steps.
Suppose that this is violated, and $\sigma^{+}$ is found as the post-state of a counterexample to induction (\cref{ln:itp-termmin:cti}). Then, should the algorithm attempt to drop the literal corresponding to $p_i$ in~\cref{ln:itp-termmin:bmc-gen}, the BMC check would pass: neither $\sigma^{+}$ nor $\sigma^{-}$ can reach $\Bad$ in $k$ steps. This would lead to inadvertently adding $\sigma^{-}$ to the candidate invariant, violating $\varphi \implies I$ (in contrast to the guarantee of~\Cref{lem:itp-underapprox}), which may hinder convergence.
$\sigma^{+}$ may happen not to be a possible counterexample to induction in any intermediate iteration and such a problematic scenario could not materialize. We note however that due to generalization, the counterexamples to induction the algorithm finds can be states that are not reachable in a small number of steps, or indeed at all reachable; in essence, the idea behind the fence condition is to ensure that all the discovered counterexamples continue to come from $I$ (\Cref{lem:itp-underapprox}) even when assuming
that \emph{any} $\sigma^{+} \models I$ can be discovered this way. %
\end{remark}  %
\subsection{Inference of Monotone Invariants}
\label{sec:itp-monotone}
In this section we
prove that~\Cref{alg:itp-termmin} converges in $m$ iterations when a backwards $k$-fenced, monotone DNF invariant with $m$ terms exists.
Monotonicity is essential: even if a short DNF invariant exists, the fence condition guarantees that each iterations learns an underapproximation of the invariant (\Cref{lem:itp-underapprox}), but exponentially many iterations could be required before the algorithm converges; we show that this cannot happen with monotone DNF invariants.
\begin{definition}[Monotone DNF]
A formula $\psi \in \mondnf{m}{}$ if it is in DNF with $m$ terms, and variables appear only positively.
\end{definition}

\Cref{eq:overview-parity-example-inv} is an example of a monotone DNF invariant.
The importance of monotonicity for our purpose stems from a classical property (see~\cite[e.g.][]{DBLP:journals/cacm/Valiant84}), concerning \emph{prime implicants}:
\begin{definition}[Prime Implicant]
A term $\land s$, where $s$ is a set of literals, is a \emph{prime implicant} of a formula $\psi$ if $(\land s) \implies \psi$, but for every $\ell \in s$, $\left(\land \left(s \setminus \set{\ell}\right)\right) \not\implies \psi$. It is \emph{non-trivial} if $\land s \not\equiv \false$.
\end{definition}
\begin{claim}[Folklore]
\label{thm:monotone-prime-implicant}
A non-trivial prime implicant of a monotone DNF formula $\psi$ is a term of $\psi$.
\end{claim}
\toolong{
\begin{proof}
Let $s$ be an implicant of $\psi$. We argue that there is a term $\hat{s}$ of $\psi$ such that $\hat{s} \subseteq s$ (as sets of literals). If $s \neq \hat{s}$, it is not a prime implicant, since dropping any $\ell \in s \setminus \hat{s}$ would also yield an implicant of $\hat{s}$ (since is a conjunction of literals) and thus of $\psi$ (since it is a disjunction of terms).

Let $v$ be the valuation that assigns $\true$ to the literals in $\psi$ and $\false$ to others; it is well-defined because $p,\neg p$ cannot both appear in $s$.
Then $v \models s$. Since $s$ is an implicant of $\psi$, also $v \models \psi$. This is a disjunction, so there is some term $\hat{s}$ of $\psi$ such that $v \models \psi$. Since $\hat{s}$ is a conjunction, $v \models \ell$ for every $\ell \in \hat{s}$, and this occurs only when $\hat{s} \subseteq s$.
\end{proof}
}

We now leverage this for the analysis of~\Cref{alg:itp-termmin}.
The idea is that the generalizations $d$ the algorithm produces (\crefrange{ln:itp-termmin:gen-start}{ln:itp-termmin:gen-end}) are prime implicants of $I$, hence each produces a new term of $I$ when $I$ is monotone.
\begin{theorem}
\label{thm:monotone-inference-efficient}
\label{lem:itp:prime-implicant}
	Let $(\Init,\tr,\Bad)$ be a transition system and $k\in\mathbb{N}$.
	If there is an inductive invariant $I \in \mondnf{m}{}$ that is backwards $k$-fenced, then
	\ourinterpolationalgname($\Init,\tr,\Bad,k$) converges to an inductive invariant in $\bigO(m)$ inductiveness checks, $\bigO(m n)$ $k$-BMC checks, and $\bigO(m n)$ time.\begin{changebar}\footnote{As explained in~\Cref{sec:prelim}, we measure time complexity modulo the SAT solver---each inductiveness and BMC query
is counted as one step.}\end{changebar}
\end{theorem}
\begin{proof}
	We first show the generalizations $d$ are prime implicants of $I$ (this holds even when $I$ is not monotone).
	That $d \implies I$ was established in the proof of~\Cref{lem:itp-underapprox}. Suppose, for the sake of contradiction, that for some literal $\ell \in d$ it holds that $d \setminus \set{\ell} \implies I$ as well. At some point the algorithm attempted to drop $\ell$; let $\tilde{d}$ be the term the algorithm considered at that point. Since we only drop literals afterwards, $d \subseteq \tilde{d}$, and hence $d \setminus \set{\ell} \subseteq \tilde{d} \setminus \set{\ell}$; since these are conjunctions, this means that $\tilde{d} \setminus \set{\ell} \implies d \setminus \set{\ell}$. Since no state in $I$ reaches $\Bad$ in any number of steps, $\bmc{\tr}{d \setminus \set{\ell}}{k} \cap \Bad = \emptyset$, and in particular $\bmc{\tr}{\tilde{d} \setminus \set{\ell}}{k} \cap \Bad = \emptyset$. But according to this check, the algorithm would have dropped $\ell$, which is a contradiction to $\ell \in d$.

	We turn to the overall analysis. By~\Cref{lem:no-fail} the algorithm does not need to restart, and converges to an inductive invariant $\varphi \implies I$. In every iteration the algorithm disjoins to $\varphi$ a term of $I$, because it is a prime implicant of $I$ and $I$ is monotone and by~\Cref{thm:monotone-prime-implicant}. %
	Furthermore, each iteration produces a new term of $I$, since before it, $\sigma' \not\models \varphi$, but $\sigma' \models d$. Thus, after at most $m$ iterations the algorithm must have added to $\varphi$ all the terms of $I$, so $I \implies \varphi$. At this point, from~\Cref{lem:itp-underapprox}, $I \equiv \varphi$ and the algorithm terminates.
	Each iteration performs one inductiveness check in~\cref{ln:itp-termmin:ind-check,ln:itp-termmin:cti}, one $k$-BMC in~\cref{ln:itp-termmin:check-no-overapprox}, and another $k$-BMC checks in~\cref{ln:itp-termmin:bmc-gen} for each of the $n$ literals.
\end{proof}

\begin{remark}
	\Cref{thm:monotone-inference-efficient} has implications also for systems that do not satisfy its requirements. If the algorithm has not converged to an invariant in the number of steps specified in~\Cref{thm:monotone-inference-efficient}, this is a \emph{witness} that an invariant satisfying the theorem's conditions \emph{does not exist}. This may indicate a bug rendering the system unsafe, or that an invariant exists but is not $k$-fenced, not monotone, or is too long.
\end{remark}

\begin{changebar}
\begin{remark}
What happens when~\Cref{alg:itp-termmin} uses a bound $k'$ that is too small, even though a $k$-fenced invariant, with $k > k'$, exists?
When using the smaller bound $k'$ naively, the algorithm might overgeneralize to beyond the gfp and fail to find an invariant, or converge to a different invariant that does not admit a short representation in an exponential number of steps. The polynomial bound guaranteed from the larger $k$ can be recovered by increasing the bound once the number of steps surpasses a predefined polynomial, or by running all possible bounds in parallel / diagonally until one instance finds an invariant.
\end{remark}
\end{changebar}

\para{Inference of (anti)monotone CNF invariants}
\label{sec:dual-interpolation}
An efficiency result for inferring (anti)monotone \emph{CNF} invariants that satisfy the \emph{forward} fence condition follows, through the duality discussed in~\Cref{sec:forward-backward-duality}.
\begin{definition}
A formula $\psi \in \moncnf{m}{}$ if it is in CNF with $m$ clauses, and variables appear only negatively.
\end{definition}
\begin{definition}[Forwards $k$-Fenced]
\label{def:fence-forwards}
For a transition system $(\Init,\tr,\Bad)$, an inductive invariant $I$ is \emph{forwards $k$-fenced} for $k \in \mathbb{N}$ if $\boundarypos{I} \subseteq \bmc{\tr}{\Init}{k}$.
\end{definition}
More explicitly, an invariant $I$ is forwards $k$-fenced if every state in $I$ that has a Hamming neighbor in $\neg I$ is reachable from $\Init$ in at most $k$ steps.
From~\Cref{thm:monotone-inference-efficient} we obtain:
\begin{corollary}
\label{thm:monotone-inference-efficient-cnf}
	Let $(\Init,\tr,\Bad)$ be a transition system and $k \in \mathbb{N}$.
	If there exists an inductive invariant $I \in \moncnf{m}{}$ that is forwards $k$-fenced, then
	$\dualourinterpolationalgname = \ourinterpolationalgname^{*}$($\Init,\tr,\Bad,k$) converges to an inductive invariant in $\bigO(m)$ inductiveness checks, $\bigO(m n)$ $k$-BMC checks, and $\bigO(m n)$ time.
\end{corollary}
The code of the dual algorithm appears in~\refappendix{sec:dual-interpolation}.

\begin{remark}
\label{rem:unate-itp}
\Cref{thm:monotone-inference-efficient,thm:monotone-inference-efficient-cnf} can be extended to the case of \emph{unate}~\cite{DBLP:journals/jacm/AngluinHK93} formulas---DNF/CNF formulas where each variable (separately) appears only negatively or only positively---when the backwards/forwards fence condition holds, using essentially the same proof as in~\Cref{thm:monotone-prime-implicant,thm:monotone-inference-efficient}.
\end{remark}

\para{\Cref{alg:itp-termmin} beyond monotone invariants}
What happens when we try to apply~\Cref{alg:itp-termmin} to infer invariants that satisfy the fence condition but are not monotone (or unate)?
By~\Cref{lem:itp-underapprox}, the algorithm would converge to an inductive invariant---but this make take an exponential number of iterations. In~\Cref{lem:itp:prime-implicant} we have shown that each iteration produces a prime implicant of $I$.
Unfortunately, the number of prime implicants can be exponential~\cite{DBLP:journals/dm/ChandraM78,DBLP:journals/siamdm/SloanST08}.
Worse, \citet{DBLP:journals/ml/AizensteinP95} have shown, in the context of exact learning, that there are cases when there is a unique, short representation of $I$ as a disjunction of prime implicants, but greedily collecting prime implicants cannot escape sifting through exponentially many additional prime implicants (which are subsumed by the prime implicants in the ``right'' representation).
\TODO{there's some generalization in~\cite{DBLP:journals/ml/AizensteinP95}, distance 2}
\section{Inference Beyond Monotone Invariants}
\label{sec:beyond-monotone}
In this section we transcend the class of monotone invariants, and obtain efficiency results for inferring inductive invariants in almost-monotone DNF (as well as additional classes that admit a small monotone basis), based on the $\Lambda$-learning algorithm by~\citet{DBLP:journals/iandc/Bshouty95}.
We first present the invariant inference algorithm and a self-contained proof of efficient convergence relying on the fence condition. In~\Cref{sec:bshouty-analysis-queries} we obtain an alternative proof by a transformation that
can ``simulate'' the original algorithm through the fence condition.
\begin{algorithm}[H]
\caption{Invariant inference using a known monotone basis, building on~\cite{DBLP:journals/iandc/Bshouty95}}
\label{alg:bshouty-known-basis}
\vspace{-0.5cm}
\begin{multicols}{2}
\begin{algorithmic}[1]
\begin{footnotesize}
\State Assuming a known basis $\set{a_1,\ldots,a_t}$ (\Cref{def:monotone-basis})
\Procedure{\bshoutyinferencealg}{$\Init$, $\tr$, $\Bad$, $k$}
	\State $H_1,\ldots,H_t \gets \false$
	\While{$H$ $=$ $\bigwedge_{i=1}^{t}{H_i}$ not an inductive inductive}
		\State \textbf{let} $\sigma'$ s.t.\ $(\sigma,\sigma') \models H \land \tr \land \neg H'$ 
		\\ \qquad \qquad \qquad \ \ \ or \ \ \ \ $\sigma' \models \Init \land \neg H$ $\label{ln:bhsouty:positive-example}$
		\For{$i=1,\ldots,t$}
			\If{$\sigma' \not\models H_i$}
				\State $d$ $\gets$ \Call{MonGenBmc}{$\sigma'$, $a_i$, $\Init$, $\tr$, $k$}
				\State $H_i \gets H_i \lor d$
			\EndIf
		\EndFor
	\EndWhile
	\State \Return $\bigwedge_{i=1}^{t}{H_i}$
\EndProcedure

\Procedure{MonGenBmc}{$\sigma$, $a$, $\tr$, $\Bad$, $k$}
	\If{$\bmc{\tr}{\sigma}{k} \cap \Bad \neq \emptyset$} $\label{ln:bshouty:check-no-overapprox}$
		\State \textbf{restart} with larger $k$ $\label{ln:bshouty:fail}$
	\EndIf
	\State $v \gets \sigma$; walked $\gets$ $\true$
	\While{walked}
		\State walked $\gets$ $\false$
		\For{$j=1,\ldots,n$ such that \begin{changebar}$a[p_j] \neq v[p_j]$\end{changebar}}
			\State $x \gets v[p_j \mapsto a[p_j]]$
			\If{$\bmc{\tr}{x}{k} \cap \Bad = \emptyset$} $\label{ln:bshouty:check-bmc}$ %
				\State $v \gets$ x; walked $\gets$ $\true$
			\EndIf
		\EndFor
	\EndWhile
	\State \Return{$\cubemon{v}{a}$}
\EndProcedure
\end{footnotesize}
\end{algorithmic}
\end{multicols}
\vspace{-0.4cm}
\end{algorithm}
\vspace{-0.4cm} 
\subsection{Inference with a Monotone Basis}
\label{sec:a-monotone-background}
In this section we present~\Cref{alg:bshouty-known-basis}, an algorithm for inferring inductive invariants with a known monotone basis. Choosing an appropriate basis produces the algorithm for almost-monotone DNF, and is described below in~\Cref{sec:classes-known-basis}. %
We use, and present along the way, notions from the monotone theory~\cite{DBLP:journals/iandc/Bshouty95}.\footnote{Our presentation differs somewhat from the original. \iflong\else Omitted proofs appear in the supplementary materials.\fi}
The main theme is \emph{variable translation}~\cite[see e.g.][]{wiedemann1987hamming}, which intuitively replaces a variable with its negation. Clearly, a unate formula \begin{changebar}(a formula where each variable appears either only negatively or only positively---see~\Cref{rem:unate-itp})\end{changebar} can be made monotone in this way. The key to transcending unate formulas is to consider \emph{multiple} translations and the ``monotonization'' of the formula according to each of these\sharon{can you relate this to the formulas being handled: e.g., they are conjunctions of such monotonizations?}\yotamsmall{yes, if we say that a monotonization is an overapproximation. don't you feel it's vague in exactly the right amount? :)}. %

\begin{definition}[Translation]
A \emph{translation} is represented by a valuation $a$. Intuitively, if $a[p_i]=\true$ then the translation defined by $a$ replaces $p_i$ with its negation; if $a[p_i]=\false$, it does not. (Semantically, applying $a$ to a state $x$ yields the state obtained by bitwise xor, $x \oplus a$.)
\end{definition}

\begin{definition}[$a$-Monotonicity]
Let $v,x$ be valuations (states) and $a$ a translation.
We define the partial order $v \leq_a x$ to mean that $x$ disagrees with $a$ on all variables on which $v$ disagree with $a$. That is, for every variable $p_i$ such that $v[p_i] \neq a[p_i]$ also $x[p_i] \neq a[p_i]$. (For example, for $a = \vec{0}$, the translation that assigns $\false$ to all variables, $v \leq_{\vec{0}} x$ means that in each variable where $v$ is $\true$, so is $x$.)
A formula $\psi$ is $a$-monotone if
$
\forall v \leq_a x. \ v \models \psi \mbox{ implies } x \models \psi.
$
\end{definition}
For example, monotone formulas are $\vec{0}$-monotone.
When $\psi$ is in DNF, $\psi$ is $a$-monotone if every variable $p$ appears in $\psi$ in polarity according to $\neg a[p]$:
$p$ appears only positively if $a[p]=\false$ and only negatively if $a[p]=\true$.
In other words, $\psi$ in DNF is $a$-monotone if interchanging $p,\neg p$ whenever $a[p]=\true$ results in a formula that is monotone DNF (per the standard definition).

In general, a formula may not be $a$-monotone for any $a$; but it can always be expressed as a conjunction of formulas, each monotone w.r.t.\ some translation. A set of translations that suffices for this is called a basis:
\begin{definition}[Monotone Basis]
\label{def:monotone-basis}
A \emph{basis} is a set of translations $\set{a_1,\ldots,a_t}$.
It is a basis \emph{for a formula} $\psi$ if
there exist clauses $c_1,\ldots,c_s$ such that $\psi \equiv c_1 \land \ldots \land c_s$ and for every $1 \leq i \leq s$ there exists $1 \leq j \leq t$ such that $c_i$ is $a_j$-monotone.
That is, in some CNF representation of $\psi$, every clause is monotone w.r.t.\ at least one translation from the basis. %
\end{definition}

\Cref{alg:bshouty-known-basis} infers an invariant under the assumption that a known set of translations is a basis for the invariant. (We address the choice of the basis in~\Cref{sec:classes-known-basis}.)
The main idea of~\Cref{alg:bshouty-known-basis} is to think about the desired invariant $I$ as a conjunction $\bigwedge_{i=1}^{t}{H_i}$ where each $H_i$ is $a_i$-monotone, and infer the $H_i$ formulas.
\yotamsmall{omitting:}%
It is immediate from the definition of a basis that a formula can be expressed in this way by letting each conjunct $H_i$ be the CNF formula that consists of the $a_i$-monotone clauses in the CNF representation. However, the number of clauses in the CNF representation, $s$, may be much larger than the number of elements in the basis, $t$.
It is not immediate that there exist \emph{short} $a_i$-monotone formulas $H_i$ that achieve this, and that they can in fact be inferred efficiently.
The main concept that enables this is the \emph{least $a_i$-monotone overapproximation} of a formula, which the algorithm infers in DNF (and not CNF), resulting in a complexity that depends on $t$, but not on $s$.

Given a formula $\varphi$ and a translation $a$, the \emph{least $a$-monotone overapproximation} of $\varphi$ is the function\footnote{Every Boolean function is expressible by a propositional formula, so we use the function interchangeably with a formula representing it (chosen arbitrarily).} $\monox{\varphi}{a}$ defined by
\begin{equation*}
	x \models \monox{\varphi}{a} \mbox{ iff } \exists v. \ v \leq_a x 	\land	 v \models \varphi.
\end{equation*}
This is the least $a$-monotone function such that $\varphi \implies \monox{\varphi}{a}$. %

Each $H_i$ will attempt to learn $\monox{I}{a_i}$. The following lemma shows that when this is achieved for all the $a_i$'s in the basis, the original formula is indeed their conjunction:
\begin{lemma}[\citet{DBLP:journals/iandc/Bshouty95}, Lemma 3]
\label{lem:basis-conj-monox}
If $\set{a_1,\ldots,a_t}$ is a basis for $\varphi$, then $\varphi \equiv \bigwedge_{i=1}^{t}{\monox{\varphi}{a_i}}$.
\end{lemma}
\toolong{
\begin{proof}
First, $\varphi \implies \monox{\varphi}{a_i}$ directly from the definition, and so $\varphi$ implies also their conjunction.
For the other direction, let $c_1 \land \ldots \land c_s$ be a CNF representation of $\varphi$ where each $c_i$ is $a_j$-monotone for some $j$.
We first argue that $a_j \not\models c_i$.
Let $\ell$ be a literal in the disjunction that is $c_i$. We show that $a_j \not\models \ell$, because if $p$ is the literal in $\ell$, then $a[p]=\false$ mandates $\ell = p$, and $a[p]=\true$ mandates $\ell=\neg p$. For otherwise, a valuation $v$ that falsifies the other literals in $c_i$ would have $v[p \mapsto a_j[p]] \models c_j$ but $v[p \mapsto \neg a_j[p]] \not\models c_j$, which is a contradiction to $c_i$ being $a_j$-monotone.

Let $x \models \bigwedge_{i=1}^{t}{\monox{\varphi}{a_i}}$. We want to prove that $x \models c_i$ for every $i$.
For every $c_i$, $x \models \monox{\varphi}{a_j}$. It suffices to show $\monox{\varphi}{a_j} \implies c_i$.
Let $y \models \monox{\varphi}{a_j}$. Then there is $v \leq_{a_j} y$ such that $v \models \varphi$, and in particular $v \models c_i$. But $c_i$ is a disjunction, and $a_j \not\models c_i$, so the variable $p_\ell$ that appears in $\ell \in c_i$ must disagree with $a_j$; $y$ thus also disagree with $a_j$ on $p_\ell$ (since $v \leq_{a_j} y$), namely, agree with $v$. This is for every $\ell \in c_i$, so $y$ and $v$ agree on all the variables in $c_i$, and hence $y \models c_i$ as well.
\end{proof}
}

A key feature of the $\monox{\varphi}{a_i}$'s is that they are guaranteed to have a short DNF representation, provided that $\varphi$ has one. This is established by the following lemmas. %
\begin{lemma}[\citet{DBLP:journals/iandc/Bshouty95}, Lemma 1(7)]
\label{lem:term-monotonization}
For a term $t$, $\monox{t}{a} \equiv \bigwedge \set{\ell \in t \, | \, a \not\models \ell}$.
\end{lemma}
\toolong{
\begin{proof}
Denote $\psi = \bigwedge \set{\ell \in t \, | \, a \not\models \ell}$.

Let $x \in \monox{t}{a}$. Then there is $v \models t$ such that $v \leq_a x$. Let $\ell \in t$; $v \models \ell$. If $a \not\models \ell$, $v,a$ must disagree on the variable $p_i$ which appears in $\ell$; therefore, $x,v$ must agree on $p_i$, and hence also $x \models \ell$. This proves $\monox{t}{a} \implies \psi$.

For the other direction, let $x \models \psi$. Let $v$ be obtained from $x$ by setting every variable $p_i$ that does not appear in $\psi$ to disagree with the corresponding value in $a$; then $v \leq_a x$. Now $v \models \psi$ (since these variables do not appear in $\psi$), and, furthermore, $v \models \ell$ for every $\ell \in t$ that was dropped from $t$ to $\psi$, because $v$ disagrees with $a$ on those literals, which are those that $a \not\models \ell$. Overall, $v \models t$, which implies $x \models \monox{t}{a}$.
\end{proof}
}

\begin{lemma}[\citet{DBLP:journals/iandc/Bshouty95}, Lemma 1(7)]
\label{lem:bshouty-mon-mindnf}
Let $\varphi = t_1 \lor \ldots \lor t_m$ in DNF. Then the monotonization $\monox{\varphi}{a} \equiv \monox{t_1}{a} \lor \ldots \monox{t_n}{a}$ which is a DNF with $m$ terms.\sharon{for the sake of shortening, can inline \Cref{lem:term-monotonization} into this lemma. Since later we refer to it, can even have itemize in the lemma and refer to the specific item, something like:
\begin{inparaenum}[(1)]
\item For a term $t$, $\monox{t}{a} \equiv \bigwedge \set{\ell \in t \, | \, a \not\models \ell}$.
\item For a DNF formula $\varphi = t_1 \lor \ldots \lor t_m$, $\monox{\varphi}{a} \equiv \monox{t_1}{a} \lor \ldots \monox{t_n}{a}$ which is a DNF with $m$ terms.
\end{inparaenum}}
\end{lemma}
\toolong{
\begin{proof}
That it is a DNF formula follows from ~\Cref{lem:term-monotonization}.
More generally, $\monox{\psi_1 \lor \psi_2}{a} \equiv \monox{\psi_1}{a} \lor \monox{\psi_2}{a}$:
Let $x \models \monox{\psi_1 \lor \psi_2}{a}$. Then there is $v \models \psi_1 \lor \psi_2$ such that $x \leq_a x$. If $v \models \psi_1$, by definition we must have $x \models \monox{\psi_1}{a}$ and in particular $x \models \monox{\psi_1}{a} \lor \monox{\psi_2}{a}$; similarly for $\psi_2$. This shows $\monox{\psi_1 \lor \psi_2}{a} \implies \monox{\psi_1}{a} \lor \monox{\psi_2}{a}$.
As for the other direction, let $x \models \monox{\psi_1}{a} \lor \monox{\psi_2}{a}$. Without loss of generality, assume $x \models \monox{\psi_1}{a}$. Then there is $v \models \psi_1$, and in particular $v \models \psi_1 \lor \psi_2$, such that $v \leq_a x$. So we must have $x \models \monox{\psi_1 \lor \psi_2}{a}$.
\end{proof}
}

Our goal now is to gradually infer $\monox{I}{a_i}$, despite $I$ being unknown.
This is done by iteratively obtaining states that ought to be added to the current hypothesis $\bigwedge_{i=1}^{t}{H_i}$ (\cref{ln:bhsouty:positive-example}). Such a state must be added to every $H_i$ that does not yet include it.
This is done by adding (disjoining) an $a_i$-monotone term to $H_i$.
The term we add (ignoring generalization at this point) is the monotone cube:
\begin{definition}[Monotone Cube]
For a state $v$ and a translation $a$, the \emph{monotone cube} of $v$ w.r.t.\ $a$ is the conjunction of all $a$-monotone literals that hold in $v$:
\begin{equation*}
	\cubemon{v}{a} = \bigwedge{\set{p_i \ | \ v[p_i]=\true, \, v[p_i] \neq a[p_i]}} \land
					 \bigwedge{\set{\neg p_i \ | \ v[p_i]=\false, \, v[p_i] \neq a[p_i]}}.
\end{equation*}
(In fact, $\cubemon{v}{a} \equiv \monox{\cube{v}}{a}$, as can be easily seen from~\Cref{lem:term-monotonization}.)
\sharon{add here? "Note that $\cubemon{v}{a} \equiv \monox{\cube{v}}{a}$". If we want it, this seems like a suitable place, but is it helpful?}\yotamsmall{added, not sure if it's useful, but I guess people might wonder}
\end{definition}
The monotone cube includes more states than the original, but the following lemma shows that this cannot overgeneralize beyond $\monox{\varphi}{a_i}$:
\begin{lemma}
\label{lem:mon-cube}
If $v \models \varphi$ then $\cubemon{v}{a} \implies \monox{\varphi}{a}$.
\end{lemma}
\toolong{
\begin{proof}
The set of valuations satisfying $\cubemon{v}{a}$ is the set of valuations $x$ such that $v \leq_a x$.
\end{proof}
}
Adding the monotone cube is thus ``safe'', but may converge to $\monox{I}{a_i}$ too slowly. \Cref{lem:bshouty-mon-mindnf} guarantees the existence of a short DNF representation, but the monotone cube might be too large (include too few states) and not be a term in this representation.
To achieve fast convergence we want to learn an actual, syntactic, term of $\monox{I}{a_i}$ whenever we add a term to $H_i$.
The mechanism that produces such terms is \emph{generalization}, in~\cref{ln:bshouty:check-bmc}, by means of minimization. %
The idea is that the more literals on which the state $v$ we will be adding to $H_i$ agrees with $a_i$, the smaller the conjunction in the $\cubemon{v}{a_i}$ is. In fact, this minimization successfully achieves an actual term of $\monox{I}{a_i}$. This is established in the next lemma, akin to~\Cref{thm:monotone-prime-implicant} in the purely monotone case.
\begin{lemma}[\citet{DBLP:journals/iandc/Bshouty95}, Proposition A + Lemma 1(1)]
\label{lem:bshouty-prop-a}
A state $x$ is \emph{$a$-minimal positive} for $\varphi$ if $x \models \varphi$ and for every $i$ such that $x[p_i] \neq a[p_i]$ it holds $x[p_i \mapsto a[p_i]] \not\models \varphi$.
Let $x$ be $a$-minimal positive for a formula $\varphi$ in DNF.
Then there is a term $t$ of $\varphi$ such that $\monox{t}{a} \equiv \cubemon{x}{a}$.
\end{lemma}
\begin{proof}
Since $x \models \varphi$ which is in DNF, there is a term $t$ of $\varphi$ such that $x \models t$.
From~\Cref{lem:mon-cube}, $\cubemon{x}{a} \implies \monox{t}{a}$.
For the other direction, by~\Cref{lem:term-monotonization}, we need to show that the conjunction $\monox{t}{a}$ includes all the conjuncts in $\cubemon{x}{a}$. To this end, let $p_i$ be such that $x[p_i] \neq a[p_i]$; we need to show that $p_i$ is a literal of $t$ if $x[p_i]=\true$, and $\neg p_i$ is a literal of $t$ if $x[p_i] = \false$.
Suppose otherwise. Then $x[p_i \mapsto \neg x[p_i]]$ also satisfies $t$. But then $x[p_i \mapsto a[p_i]] = x[p_i \mapsto \neg x[p_i]] \models \varphi$, in contradiction to the premise.
\end{proof}

\para{The algorithm: \bshoutyinferencealgplain}
We now collect the ideas from above and describe the algorithm (\Cref{alg:bshouty-known-basis}) for inferring invariants that admit a known monotone basis, based on the backwards fence condition.
Assuming a known basis $\set{a_1,\ldots,a_t}$ for the target (unknown) invariant $I$, the algorithm maintains a sequence $H_1,\ldots,H_t$, where each $H_i$ is an $a_i$-monotone DNF formula. Each $H_i$ is gradually increased until it is $\monox{I}{a_i}$ (unless an invariant is found earlier). When each $H_i$ attains this limit, $I \equiv \bigwedge_{i=1}^{t}{H_i}$ and we are done. In a sense, the algorithm combines multiple instances of the inference procedure appropriate for the monotone case (\Cref{alg:itp-termmin}), each for learning an $a_i$-monotonization of $I$.

Each $H_i$ starts from $\false$. When a state that ought to be added to the current hypothesis $\bigwedge_{i=1}^{t}{H_i}$ is found (\cref{ln:bhsouty:positive-example}), each $H_i$ that does not include it is increased by adding a new term.
In order to learn an actual, syntactic term of $\monox{I}{a_i}$, the algorithm
gradually flips bits in the state that disagree with $a_i$, and \emph{heuristically} checks whether the new state should still be included in the invariant by performing \emph{bounded model checking} (\cref{ln:bshouty:check-bmc}). %
When the fence condition holds, this mimics the procedure from~\Cref{lem:bshouty-prop-a}; %
this is important for the algorithm's efficiency, below. %
Before embarking on efficiency guarantees of this algorithm, we note that the algorithm is always sound (even when $I$ is not $k$-fenced or $\set{a_1,\ldots,a_t}$ is not a basis for $I$), because it checks that $H$ is inductive before returning; if $H$ is not inductive, the algorithm continues to increase $H_i$'s until $H$ includes a state that reaches $\Bad$ and the algorithm reaches failure.

Our main theorem for this algorithm is that when the $k$-fenced condition holds, the algorithm can \emph{efficiently} learn every formula for which $\set{a_1,\ldots,a_t}$ is a basis:
\begin{theorem}
\label{thm:bshouty-known-efficiency}
	Let $(\Init,\tr,\Bad)$ be a transition system, and $k \in \mathbb{N}$.
	If there exists an inductive invariant $I$ that is backwards $k$-fenced, $I \in \dnf{m}{}$, and $\set{a_1,\ldots,a_t}$ is a monotone basis for $I$ (\Cref{def:monotone-basis}), then
	\bshoutyinferencealg($\Init,\tr,\Bad,k$) converges to an inductive invariant in $\bigO(m \cdot t)$ inductiveness checks, $\bigO(m \cdot t \cdot n^2)$ $k$-BMC checks, and $\bigO(m \cdot t \cdot n^2)$ time.
\end{theorem}
\label{sec:bshouty-direct}
\begin{proof} %
Our main claim is that $H_i \implies \monox{I}{a_i}$ and $H_i$ is $a_i$-monotone (for every $i$).
From this it would follow that $\bigwedge_{i=1}^{t}{H_i} \implies I$, because $I = \bigwedge_{i=1}^{t}{\monox{I}{a_i}}$ from the basis assumption and~\Cref{lem:basis-conj-monox}. From this it would follow that the counterexample is always positive, $\sigma' \models I$. {This implies that $\bmc{\tr}{\sigma'}{k} \cap \Bad = \emptyset$,
so the algorithm does not fail (\cref{ln:bshouty:fail}).}

Initially, the claim holds trivially. Consider an iteration. From the induction hypothesis, and as above, $\sigma' \models I$.
Thus generalization begins with $x = \sigma' \models I$.
{We argue by induction on the steps of generalization that $x \models I$. In each step, we move from a state $x$ to a Hamming neighbor state $x'$ s.t.\ $\bmc{\tr}{x'}{k} \cap \Bad = \emptyset$. By the induction hypothesis and the premise that $I$ is $k$-backwards fenced, also $x' \models I$, which concludes this induction.
The final $x$ in generalization thus has $x \models I$.} %
Therefore, by~\Cref{lem:mon-cube}, $\cubemon{x}{a_i} \implies \monox{I}{a_i}$ for every $i$.
The claim follows.

It remains to argue that after at most $m \cdot t$ iterations the algorithm converges to $\bigwedge_{i=1}^{t}{H_i} \equiv I$ (unless it terminates earlier with an inductive invariant), because every call to generalization takes at most $O(n^2)$ $k$-BMC queries.
Indeed, every iteration adds at least one term to at least one $H_i$. For the $x$ that produces the term, $x \models I$, as above, but for every $p$ where $x[p] \neq a[p]$ {we have $\bmc{\tr}{x[p \mapsto a[p]]}{k} \cap \Bad \neq \emptyset$, and in particular $x[p \mapsto a[p]] \not\models I$}. Using \Cref{lem:bshouty-prop-a}, $\cubemon{x}{a_i}$ is a term of the DNF representation of $\monox{I}{a_i}$ from~\Cref{lem:bshouty-mon-mindnf}.
By~\Cref{lem:bshouty-mon-mindnf} this representation has $m$ terms. Overall we need at most $t \cdot m$ iterations.
\end{proof}

\subsection{Choosing a Monotone Basis}
\label{sec:classes-known-basis}
Some important classes of formulas have a known basis that the algorithm can use.
The class of \emph{$r$-almost-monotone DNF} is the class of DNF formulas with at most $r$ terms which include negative literals.
The set of all translations with at most $r$ variables assigned $\true$ is a basis for this class~\cite{DBLP:journals/iandc/Bshouty95}.
When $r = \bigO(1)$, the size of this basis is polynomial in $n$.
This is a basis because, when converting an almost-monotone DNF formula to CNF form, every clause has at most $r$ negative literals,
which is $a$-monotone for the translation $a$ which assigns %
$\true$ to %
these variables only. %
Another interesting class with a known base of size polynomial in $n$ is the class of (arbitrary) DNF formulas with $\bigO(\log n)$ terms, although the construction is less elementary~\cite{DBLP:journals/iandc/Bshouty95}.

Applying~\Cref{thm:bshouty-known-efficiency} with the known basis for $r$-almost-monotone DNF yields:
\begin{corollary}
\label{cor:almost-monotone-dnf}
	Let $(\Init,\tr,\Bad)$ be a transition system, $k \in \mathbb{N}$, and $r = \bigO(1)$.
	If there exists an inductive invariant $I$ that is backwards $k$-fenced, and $I$ is $r$-almost-monotone DNF with $m$ terms, then
	\bshoutyinferencealg($\Init,\tr,\Bad,k$) with an appropriate basis converges to an inductive invariant in $\textit{poly}(m \cdot n)$ inductiveness checks, $\textit{poly}(m \cdot n)$ $k$-BMC checks, and $\textit{poly}(m \cdot n)$ time.
\end{corollary}

A dual result for $r$-almost (anti)monotone CNF invariants, which are CNF formulas with at most $r$ clauses that include positive literals, is as follows:
\begin{corollary}
\label{cor:almost-monotone-cnf}
	Let $(\Init,\tr,\Bad)$ be a transition system, $k \in \mathbb{N}$, and $r = \bigO(1)$.
	If there exists an inductive invariant $I$ that is forwards $k$-fenced, $I$ is $r$-almost-(anti)monotone CNF with $m$ clauses, then
	$\bshoutyinferencealg^{*}$($\Init,\tr,\Bad,k$) %
with an appropriate basis converges to an inductive invariant in $\textit{poly}(m \cdot n)$ inductiveness checks, $\textit{poly}(m \cdot n)$ $k$-BMC checks, and $\textit{poly}(m \cdot n)$ time.
\end{corollary}

\section{From Exact Learning to Invariant Inference via the Fence Condition}
\label{sec:exact-invariant-learning}
Exact learning with queries~\cite{DBLP:journals/ml/Angluin87} is one of the fundamental fields of theoretical machine learning.
In this section we show how efficient inference based on the fence condition can be understood as a manifestation of special forms of exact learning algorithms.
In~\Cref{sec:bshouty-analysis-queries} we obtain \Cref{alg:bshouty-known-basis} and an algorithm resembling~\Cref{alg:itp-termmin} by a translation from exact learning algorithms that satisfy certain restrictions. %
In particular, this provides an alternative proof of~\Cref{thm:bshouty-known-efficiency}.
In~\Cref{sec:bshouty-two-sided} we show that when both the backwards \emph{and} the forwards fence condition hold, then \emph{every} algorithm for exact learning from equivalence and membership queries can be transformed to an inference algorithm.
In particular, this proves that formulas that admit both a short CNF and a short DNF representation (neither necessarily monotone) can be inferred efficiently when the two-sided fence condition holds (\Cref{thm:bshouty-inference-cdnf}).
These transformations implement the learning algorithm's queries even though the target invariants are not known to the algorithm or to the SAT solver.
Such transformations are impossible in general~\cite{DBLP:journals/pacmpl/FeldmanISS20}, and here rely on the fence condition.

\subsection{Background: Exact Concept Learning with Queries}
\label{sec:exact-learning-background}
We begin with some background on exact concept learning with queries.
In \emph{exact concept learning}~\cite{DBLP:journals/ml/Angluin87}, the algorithm's task is to identify an unknown formula %
$\varphi$ using queries it poses to a \emph{teacher}. The most studied queries are:
\begin{itemize}
	\item \emph{Membership}: The algorithm \emph{chooses} a state $\sigma$, and the teacher answers whether $\sigma \models \varphi$; and
	\item \emph{Equivalence}: The algorithm \emph{chooses} a candidate $\theta$, and the teacher returns true if $\theta \equiv \varphi$ or a differentiating counterexample otherwise: a $\sigma$ s.t.\ $\sigma \not\models \theta, \sigma \models \varphi$ or $\sigma \models \theta, \sigma \not\models \varphi$.
\end{itemize}
The question studied is how many equivalence and membership queries suffice to correctly identify an unknown $\varphi$ from a certain (syntactical) class.

\subsection{Inference From One-Sided Fence and Exact Learning With Restricted Queries}
\label{sec:bshouty-analysis-queries}
\label{sec:restricted-queries}
\label{sec:inference-from-positive-exact}
The challenge in harnessing exact learning algorithms for invariant inference is the need to also implement the teacher, which is problematic because the algorithm does not know any inductive invariant in advance~\cite{ICELearning}. In this section we overcome this problem using the fence condition, provided that the learning algorithm satisfies some conditions.

Membership queries to an (unknown) target invariant are in general impossible to implement~\cite{ICELearning,DBLP:journals/pacmpl/FeldmanISS20}.
Even if we target the clearly-defined gfp specifically, then the query amounts to asking whether $\sigma$ can reach $\Bad$ in an unbounded number of steps, but this question is not an easier than the safety problem. If the desired $I$ is not the gfp (say, because the gfp is a complex formula), then it is even less clear how to answer the query.

Equivalence queries are also hard to implement in general~\cite{DBLP:journals/pacmpl/FeldmanISS20}: while we can determine inductiveness or find a counterexample, this may be counterexample to induction $(\sigma,\sigma')$, which is a transition, not a single state;
deciding whether to return to the learner  $\sigma$ or $\sigma'$ as a differentiating example depends on whether $\sigma \models I$, which has all the problems of a membership query above.
We will circumvent the problem of equivalence queries by considering algorithms that query only on candidates which are underapproximations of the target $I$:
\begin{lemma}[implementing positive equivalence queries]
\label{lem:infer-equivalence-positive}
	Let $(\Init,\tr,\Bad)$ be a transition system and $I$ an inductive invariant.
	Given $\theta$ such that $\theta \implies I$, it is possible to decide whether $\theta$ is an inductive invariant or provide a counterexample $\sigma \models I, \sigma \not\models \theta$, by
	\begin{itemize}
		\item checking whether there is a counterexample $\sigma' \models \Init, \sigma' \not\models \theta$ and returning $\sigma'$ if one exists; and
		\item checking whether there is a counterexample $(\sigma,\sigma') \models \theta \land \tr \land \neg\theta'$, and returning $\sigma'$ if one exists.
	\end{itemize}
Otherwise, $\theta$ is an inductive invariant.
\end{lemma}
Note that $\theta \not\equiv I$ could be an inductive invariant, which does not amount to an equivalence query \emph{per se}, but then the algorithm has already found an inductive invariant and can stop.

Our main observation here is about implementing \emph{membership} queries: that if the fence condition holds for $I$, then it is possible to efficiently implement \emph{restricted versions} of membership queries:
\begin{lemma}[implementing positive-adjacent membership queries]
\label{lem:infer-membership-positive}
	Let $(\Init,\tr,\Bad)$ be a transition system and $I$ an inductive invariant that is backwards $k$-fenced.
	Given $\sigma$ s.t.\ $\sigma \models I$ or $\sigma \in \boundaryneg{I}$, it is possible to decide whether $\sigma \models I$ using a single $k$-BMC check of whether $\bmc{\tr}{\sigma}{k} \cap \Bad = \emptyset$.
\end{lemma}
\begin{proof}
If $\sigma \models I$, it cannot reach $\Bad$ in any number of steps, $k$ in particular, and we correctly return \emph{true}. Otherwise, from the premise, $\sigma \in \boundaryneg{I}$ but $\sigma \not\models I$, so, by the fence condition, we must have $\bmc{\tr}{\sigma}{k} \cap \Bad \neq \emptyset$, and we correctly return \emph{false}.
\end{proof}

A learning algorithm that only performs such queries induces an invariant inference algorithm.
\begin{corollary}
\label{lem:simulate-positive-exact}
	Let $C$ be a class of formulas.
	Let $\A$ be an exact concept learning algorithm that can identify every $\varphi \in C$ in at most $s_1$ equivalence queries and $s_2$ membership queries.
	Assume further that when $\A$ performs an equivalence query on $\theta$, always $\theta \implies \varphi$,
	and when $\A$ performs a membership query on $\sigma$, always $\sigma \models \varphi$ or $\sigma \in \boundaryneg{\varphi}$.
	Then there exists an invariant inference algorithm that is sound (returns only correct invariants), and, furthermore, can find an inductive invariant for every transition system that admits an inductive invariant $I \in C$ that is backwards $k$-fenced using at most $s_1+1$ inductiveness checks and $s_2$ $k$-BMC checks.
\end{corollary}
\begin{changebar}
\begin{proof}
We simulate $\A$ using the equivalence queries from~\Cref{lem:infer-equivalence-positive} and the membership queries from~\Cref{lem:infer-membership-positive}. If the fence condition holds, we answer all queries correctly, perhaps except for an equivalence query on $\theta$ returning $\true$ although $\theta \not\equiv I$, but then we have already found an inductive invariant $\theta$ and can stop. 
An additional inductiveness check is used before an invariant is returned to ensure that the result is a correct inductive invariant even when the fence condition does not hold. If the latter inductiveness check fails, the algorithm returns ``failure''.
\end{proof}
\end{changebar}
Note that the resulting algorithm is sound even when the fence condition does not hold, although successful and efficient convergence is not guaranteed in this case.

\para{\Cref{alg:itp-termmin} and exact learning}
The interpolation algorithm of~\Cref{alg:itp-termmin} and its efficiency result (\Cref{thm:monotone-inference-efficient}) can almost exactly be obtained by the transformation of~\Cref{lem:simulate-positive-exact} from the exact learning algorithm for monotone DNF~\cite{DBLP:journals/ml/Angluin87}. (The code of the exact learning algorithm appears in~\refappendix{sec:exact-learning-algs}.)
The primary difference between~\Cref{alg:itp-termmin} and the algorithm resulting from the translation is that in the resulting algorithm, generalization starts with the positive literals, filtering the negative literals in advance (as opposed to~\cref{ln:itp-termmin:gen-start} in~\Cref{alg:itp-termmin}); this is reasonable for searching for monotone invariants, but may not be complete the same way the standard algorithm is (see~\Cref{cor:itp-termmin-complete}), and also does not extend to unate DNF (see~\Cref{rem:unate-itp}). Another difference is that the translated algorithm performs BMC on states with more and more $\false$ entries instead of dropping literals. The similarities outweigh the differences nonetheless.

\para{\Cref{alg:bshouty-known-basis} and exact learning}
The translation in~\Cref{lem:simulate-positive-exact} provides an alternative proof of~\Cref{thm:bshouty-known-efficiency}.
\begin{proof}[Proof of~\Cref{thm:bshouty-known-efficiency}]
\Cref{alg:bshouty-known-basis} is obtained by the transformation in~\Cref{lem:simulate-positive-exact} applied on the $\Lambda$-algorithm for exact concept learning using a known monotone basis by~\citet[][\S5]{DBLP:journals/iandc/Bshouty95} (the algorithm's code appears in~\refappendix{sec:exact-learning-algs}). The bounds on the number of inductiveness and BMC checks in our theorem matches the bounds on equivalence and membership queries of the original algorithm. It remains to argue that the $\Lambda$-algorithm satisfies the conditions of~\Cref{lem:simulate-positive-exact}. Indeed, the hypothesis is always below the true formula and counterexamples are always positive~\cite[][\S5.1.1, inductive property 1]{DBLP:journals/iandc/Bshouty95}, and membership queries are always performed after flipping one bit in a positive example.
\end{proof}

\subsection{Inference From Two-Sided Fence and Exact Learning}
\label{sec:bshouty-two-sided}
In this section we simulate arbitrary exact learning algorithms (going beyond the requirements in~\Cref{lem:simulate-positive-exact}) relying on a \emph{two-sided} fence condition.
An important example of such an exact learning algorithm %
is the CDNF algorithm by~\citet{DBLP:journals/iandc/Bshouty95}. The conditions of the transformation in~\Cref{sec:inference-from-positive-exact} do not hold because this algorithm performs equivalence queries that can return either positive or negative examples.
We first exemplify the two-sided fence condition.
\begin{example}
\label{ex:two-sided-fence-overview-example}
In the example of~\Cref{fig:parity}, the invariant in~\Cref{eq:overview-parity-example-inv} is backwards $2$-fenced (see~\Cref{sec:overview-example-fence-holds}). It is also $1$ \emph{forwards}-fenced: every state in $\boundarypos{I} = I$ is reachable in $1$ step by \code{havoc\_others}. %
\end{example}

We now show how to implement queries to the invariant using the two-sided fence condition. %
\begin{lemma}[implementing membership queries]
Let $(\Init,\tr,\Bad)$ be a transition system and $I$ an (unknown) inductive invariant that is backwards $k_1$-fenced and forwards $k_2$-fenced.
Then membership queries to $I$ can be implemented in at most $n$ queries of $k_1$-BMC and $k_2$-BMC.\footnote{
	The proof of this also implies that an invariant that is both forwards $k_1$-fenced and backwards $k_2$-fenced is unique, seeing that the implementation of the membership query for both is the same.
}%
\end{lemma}
\begin{proof}
\label{lem:infer-membership-all}
Let $\sigma$ be a state such that we want to check whether $\sigma \in I$.
Choose some known state $\sigma_0 \models \Init$, and gradually walk from $\sigma$ to $\sigma_0$, that is, in each step change one variable in $\sigma$ to match $\sigma_0$, i.e.\ $\sigma \gets \sigma[p \mapsto \sigma_0[p]]$.
In each step, check:
\begin{itemize}
	\item If $\bmc{\tr}{\Init}{k_1} \cap \set{\sigma} \neq \emptyset$ (namely $\sigma \in \bmc{\tr}{\Init}{k_1}$), return $\true$. (This is a $k_1$-BMC check.)
	\item If $\bmc{\tr}{\sigma}{k_2} \cap \Bad \neq \emptyset$ (namely $\sigma \in \bmcback{\tr}{\Bad}{k_2}$), return $\false$. (This is a $k_2$-BMC check.)
	\item Otherwise, step and recheck.
\end{itemize}
\iflong
At least one of the queries is true at some point, because $\sigma_0 \in \bmc{\tr}{\Init}{k_1}$.

Suppose $\sigma \in I$. Then in this process, as long as $\sigma$ stays in $I$, we cannot return $\false$  because states in $I$ do not reach $\Bad$ (in $k_2$ steps or more). Then either $\sigma$ always stays in $I$, in which case we return $\true$ when $\sigma$ becomes $\sigma_0$, or there is a first crossing point from $\sigma_1 \in I$ to $\sigma_2 \not\in I$, where, from the premise that $I$ is forwards $k_1$-fenced, $\sigma \in \bmc{\tr}{\Init}{k_1}$ and we return $\true$, as expected.

Suppose $\sigma \not \in I$. Then in this process, as long as $\sigma$ stays \emph{not} in $I$, we cannot return $\true$, because states in $\neg I$ are not reachable from $\Init$ (in $k_1$ steps or more). Because we end the process with $\sigma_0 \in I$ there must be a first crossing point from $\neg I$ to $I$, where, from the premise that $I$ is backwards $k_2$-fenced, $\bmc{\tr}{\sigma}{k_2} \cap \Bad \neq \emptyset$, and we return $\false$, as expected.
\else
In this process, when $\sigma \models I$, we cannot return $\false$, and when $\sigma \models \neg I$, we cannot return $\true$, because $I$ excludes states that can reach $\Bad$ (in $k_2$ steps or more) and $\neg I$ does not include states that are reachable from $\Init$ (in $k_1$ steps or more). If we first cross from $I$ to $\neg I$, we return $\true$ from $I$ being forwards $k_1$-fenced. If we first cross from $\neg I$ to $I$, similarly we return $\false$ thanks to $I$ being backwards $k_2$-fenced. Otherwise we reach $\sigma_0 \in I$ without crossing and correctly return $\true$.
\fi
\end{proof}

An equivalence query can be implemented by \begin{changebar}an\end{changebar} inductiveness check and a membership query, as noted by~\citet{DBLP:journals/pacmpl/FeldmanISS20}:
\begin{lemma}[implementing equivalence queries]
\label{lem:infer-equivalence-all}
Let $(\Init,\tr,\Bad)$ be a transition system, and $I$ an (unknown) inductive invariant that is forwards $k_1$-fenced and backwards $k_2$-fenced.
Then given $\theta$ it is possible to answer whether $\theta$ is an inductive invariant, or provide a counterexample $\sigma$ such that $\sigma \models \theta, \sigma\not\models I$ or $\sigma \not\models \theta, \sigma\models I$, using an inductiveness check, at most $n$ checks of $k_1$-BMC and $n$ of $k_2$-BMC.
\end{lemma}
\toolong{
\begin{proof}
Perform an inductiveness check of $\theta$. If there is a counterexample $(\sigma,\sigma')$ perform a membership query on $\sigma$ using~\Cref{lem:infer-membership-all}: if the result is true, return $\sigma'$ (a positive counterexample); otherwise return $\sigma$ (a negative counterexample).
\end{proof}
}

We can use these procedures to implement every exact learning algorithm from (arbitrary) equivalence and membership queries.
\begin{corollary}
\label{lem:simulate-all-exact}
	Let $C$ be a class of formulas.
	Let $\A$ be an exact concept learning algorithm that can identify every $\varphi \in C$ in at most $s_1$ equivalence queries and $s_2$ membership queries.
	Then there exists a sound invariant inference algorithm that can find an inductive invariant for every transition system that admits an inductive invariant $I \in C$ that is forwards $k_1$-fenced and backwards $k_2$-fenced using at most $s_1+1$ inductiveness checks, $n (s_1 + s_2)$ of $k_1$-BMC checks, and $n (s_1 + s_2)$ of $k_2$-BMC checks.
\end{corollary}
\toolong{
\begin{proof}
We simulate $\A$ using the equivalence queries from~\Cref{lem:infer-equivalence-all} and the membership queries from~\Cref{lem:infer-membership-all}. We answer all queries correctly, perhaps except for an equivalence query on $\theta$ returning $\true$ although $\theta \not\equiv I$, but then we have already found an inductive invariant $\theta$ and can stop. For soundness, we always verify the result before returning using an additional inductiveness check.
\end{proof}
}

Next, we demonstrate an application of \Cref{lem:simulate-all-exact} to the inference of a larger class of invariants.

\subsubsection{Inference Beyond Almost-Monotone Invariants}
\label{sec:cdnf-inference}
Earlier, we have shown that almost-monotone \emph{DNF} invariants are efficiently inferrable when the \emph{backwards} fence condition holds, and similarly for almost-monotone \emph{CNF} when the \emph{forwards} fence condition holds (\Cref{cor:almost-monotone-dnf,cor:almost-monotone-cnf}).
We now utilize \Cref{lem:simulate-all-exact} to the CDNF algorithm by~\citet{DBLP:journals/iandc/Bshouty95} %
to show that the class of invariants that can be succinctly expressed \emph{both} in DNF and in CNF (not necessarily in an almost-monotone way) can be efficiently inferred when the fence condition holds in \emph{both} directions:
\begin{theorem}
\label{thm:bshouty-inference-cdnf}
	There is an algorithm $\A$ that for every input transition system $(\Init,\tr,\Bad)$ and $k \in \mathbb{N}$,
	if the system admits an inductive invariant $I$ such that $I \in \dnf{m_1}{}$, $I \in \cnf{m_2}{}$,
	and $I$ is both backwards- and forwards- $k$-fenced, then
	$\A$($\Init,\tr,\Bad,k$) converges to an inductive invariant in $\bigO(m_1 \cdot m_2)$ inductiveness checks, $\bigO(m_1 \cdot m_2 \cdot n^3)$ $k$-BMC checks, and $\bigO(m_1 \cdot m_2 \cdot n^3)$ time.
\end{theorem}

As noted by~\citet{DBLP:journals/iandc/Bshouty95}, the class of formulas with short DNF and CNF includes the formulas that can be expressed by a small \emph{decision tree}%
\iflong
: a binary tree in which every internal node is labeled by a variable and a leaf by $\true$/$\false$, and $\sigma$ satisfies the formula if the path defined by starting from the root, turning left when the $\sigma$ assigns $\false$ to the variable labeling the node and right otherwise, reaches a leaf $\true$. The size of a decision tree is the number of leaves in the tree.
\begin{corollary}
\label{thm:bshouty-inference-decision-tree}
	There is an algorithm $\A$ that for every input transition system $(\Init,\tr,\Bad)$ and $k \in \mathbb{N}$,
	if the system admits an inductive invariant $I$ that can be expressed as a decision tree of size $m$,
	and $I$ is both backwards- and forwards- $k$-fenced, then
	$\A$($\Init,\tr,\Bad,k$) converges to an inductive invariant in $\bigO(m^2)$ inductiveness checks, $\bigO(m^2 \cdot n^3)$ $k$-BMC checks, and $\bigO(m^2 \cdot n^3)$ time.
\end{corollary}
\toolong{
\begin{proof}
A decision tree of size $m$ has a DNF representation of $m$ terms: a disjunction of terms representing the paths that reach a leaf $\true$, each is the conjunction of the variables on the path with polarity according to left/right branch. Similarly, it has a CNF representation of $m$ clauses: a conjunction of clauses which are the negations of paths that reach a leaf $\false$. Now apply~\Cref{thm:bshouty-inference-cdnf}.
\end{proof}

}
\else
; see the supplementary materials.
\fi
Similarly, when an $r$-almost-unate invariant with $\bigO(\log n)$ non-unate variables is fenced both backwards and forwards, it can be inferred by an adaptation of an algorithm by~\citet{DBLP:journals/cc/Bshouty97}. Whether this is possible based on the one-sided fence condition is an interesting question for future work.   
\section{Robustness and Non-Robustness of the Fence Condition}
\label{sec:robustness}
In this section we study the effect of program transformations on the %
$k$-fence condition, reflecting on the robustness of guaranteed convergence of the algorithms we study.

Suppose $(\Init,\tr,\Bad)$ is a transition system that admits an inductive invariant $I$ which is backwards or forwards $k$-fenced (\Cref{def:fence-backwards,def:fence-forwards}). %
Suppose that the system is modified in some way to form a new system $(\instr{\Init}, \instr{\tr}, \instr{\Bad})$ with an inductive invariant $\instr{I}$ which is derived from $I$ (often $\instr{I} = I$ itself).
The transformation $\instr{\cdot}$ %
is \emph{robust} if $\instr{I}$ is also backwards/forwards
$k$-fenced in $(\instr{\Init},\instr{\tr},\instr{\Bad})$. %
We consider several such program transformations and establish their (non)robustness.

\subsection{Robustness Under Simple Transformations}
\subsubsection{Isometries}
An isometry with respect to the Hamming distance is a permutation of the variables followed by variable translation per~\Cref{sec:a-monotone-background}~\cite[see e.g.][]{wiedemann1987hamming}.
The new $\instr{\Init},\instr{\tr},\instr{\Bad},\instr{I}$ are obtained by renaming the variables according to the permutation, and replacing variables with their negation according to the translation. The isometry preserves the Hamming distance between states. Therefore, if $I$ is backwards/forwards $k$-fenced, so is $\instr{I}$.

\subsubsection{Conjunctions and Disjunctions}
If $(\Init,\tr,\Bad_1)$ has a backwards (forwards) $k$-fenced inductive invariant $I_1$ and similarly $(\Init,\tr,\Bad_2)$ has $I_2$, then $I_1 \land I_2$ is an inductive invariant for $(\Init,\tr,\Bad_1 \lor \Bad_2)$, and it is also backwards (forwards) $k$-fenced.
Similarly,
if $(\Init_1,\tr,\Bad)$ has a backwards (forwards) $k$-fenced inductive invariant $I_1$ and similarly $(\Init_1,\tr,\Bad)$ has $I_2$, then $I_1 \lor I_2$ is an inductive invariant for $(\Init_1 \lor \Init_2,\tr,\Bad)$, and it is also backwards (forwards) $k$-fenced.
This is because for any two sets of states $S_1,S_2$, the boundary satisfies (much like in usual topology)
\begin{align*}
	&\boundarypos{S_1 \cap S_2} \subseteq \boundarypos{S_1} \cup \boundarypos{S_2}
	&\boundaryneg{S_1 \cap S_2} \subseteq \boundaryneg{S_1} \cup \boundaryneg{S_2}
	\\
	&\boundarypos{S_1 \cup S_2} \subseteq \boundarypos{S_1} \cup \boundarypos{S_2}
	&\boundaryneg{S_1 \cup S_2} \subseteq \boundaryneg{S_1} \cup \boundaryneg{S_2}
\end{align*}
and backwards (forwards) $k$-reachability is not reduced in either case.
\toolong{
\TODO{proof without lemma statement....}
\begin{proof}[Proof (boundary of union/intersection)]\sharon{didn't check}
Let $\sigma^{-} \in \boundaryneg{S_1 \cap S_2}$, so wlog.\ $\sigma^{-} \not \models S_1$ and it has a Hamming-neighbor $\sigma^{+} \models S_1 \cap S_2$. in particular $\sigma^{+} \models S_1$ and thus $\sigma^{-} \in \boundaryneg{S_1}$.
Let $\sigma^{-} \in \boundaryneg{S_1 \cup S_2}$, it has a Hamming-neighbor $\sigma^{+} \models S_1 \cup S_2$. Assume wlog.\ that $\sigma^{+} \models S_1$. Since $\sigma^{-} \not\models S_1 \cup S_2$, in particular $\sigma^{-} \not\models S_1$, and thus $\sigma^{-} \in \boundaryneg{S_1}$.
The rest of the cases are through the duality $\boundarypos{S_1 \cap S_2} = \boundaryneg{\bar{S}_1 \cup \bar{S}_2} \subseteq \boundaryneg{\bar{S}_1} \cup \boundaryneg{\bar{S}_2} = \boundarypos{S_1} \cup \boundarypos{S_2}$ and $\boundarypos{S_1 \cup S_2} = \boundaryneg{\bar{S}_1 \cap \bar{S}_2} \subseteq \boundaryneg{\bar{S}_1} \cup \boundaryneg{\bar{S}_2} = \boundarypos{S_1} \cup \boundarypos{S_2}$.
\end{proof}
}

\OMIT{
	\subsubsection{Fresh Uncorrelated Predicates}
	\label{sec:robustness-uncorrelated-predicates}
	We now consider introducing new variables, giving rise to a new vocabulary $\instr{\voc} = \voc \uplus \set{q_1,\ldots,q_s}$.
	The transformation we consider now is such that the new variables are not correlated with the old ones, nor with each other\sharon{actually I don't see what in the def below prevents correlations among the new vars}\yotamsmall{the initial state where they are all havoced}.
	In this scenario, we show that the algorithm cannot ``by mistake'' learn correlations that are not supposed to be part of the invariant.
	The new transition system does not constrain $q_1,\ldots,q_s$ on the initial state $\instr{\Init} = \Init$, nor in the bad states, $\instr{\Bad} = \Bad$. The transitions of the new system $\instr{\tr}$ may modify these variables, but allow the original transitions to proceed without altering the new variables, and do not introduce new behaviors over the old variables. %
	The transitions of the new system $\instr{\tr}$ may modify the new variables without modifying the old variables, and allow the original transitions to proceed without altering the new variables.\yotamsmall{used the suggestion} %
	Under this transformation, $\instr{I} = I$ is an inductive invariant for the new system.
	The new invariant satisfies the $k$-fenced condition, and thus inference is robust under this transformation:
	\begin{lemma}
	Let $(\Init,\instr{\tr},\Bad)$ be obtained from $(\Init,\tr,\Bad)$ by adding fresh uncorrelated variables as above. Then if $I$ is an inductive invariant that is backwards $k$-fenced for $(\Init,\tr,\Bad)$, it is also such for $(\Init,\instr{\tr},\Bad)$.
	\end{lemma}
	\begin{proof}
	In this proof we use the notation $\project{\instr{\sigma}}{\voc}$ for the projection restricting of $\instr{\sigma}$ to the vocabulary $\voc$.
	$I$ is an inductive invariant by the second condition, that guarantees that new behaviors over $\voc$ are not introduced.
	Suppose $\sigma^{+},\sigma^{-}$ are Hamming-neighbors such that $\sigma^{+} \models I, \sigma^{-} \not\models I$. Then they must differ on a variable in $\voc$, because $I$ does not mention variables from $\instr{\voc} \setminus \voc$. Therefore, $\project{\sigma^{+}}{\voc},\project{\sigma_o}{\voc}$ \sharon{if removed from above, move def of $\project{\instr{\sigma}}{\voc}$ to here}are Hamming-neighbors, and $\project{\sigma^{+}}{\voc} \models I$ while $\project{\sigma^{-}}{\voc} \not\models I$. From the premise that $I$ is backwards $k$-fenced, $\sigma^{-}$ reaches $\Bad$ in at most $k$ steps of $\tr$, let them be $\project{\sigma^{-}}{\voc}=\sigma_1,\ldots,\sigma_r \models \Bad$ where $\sigma_i,\sigma_{i+1} \models \tr$, $r \leq k$. Denote $\instr{\sigma}^i$ the state over $\instr{\voc}$ is obtained $\sigma_i$ by interpreting the variables in $\instr{\voc} \setminus \voc$ as they are interpreted in $\sigma^{-}$. From the construction, $\instr{\sigma}_i, \instr{\sigma}_{i+1} \models \instr{\tr}$\sharon{did you mean $\instr{\tr}$? why do you need the or}\yotamsmall{my bad, thanks}, so $\instr{\sigma}_1,\ldots,\instr{\sigma}_r$ is a trace of $\tr$ of length at most $k$, starting from $\sigma^{-}$ and ending in a state in $\Bad=\instr{\Bad}$. Thus $\sigma^{-}$ reaches $\instr{\Bad}$ in at most $k$ steps, as desired.
	\end{proof}

	\begin{example}
	\TODO{}\sharon{can we add something like a counter, without changing the property? i don't see how to do it without changing init, but actually, why can't we relax the transformation to allow changes to init?}
	\end{example}
}

\OMIT{
	\subsubsection{Initialization Step}
	\sharon{before the same transformation was robust for both conditions. Now there are two similar but different transformations, and each is robust for one condition}

	\sharon{I know initialization is more natural, but can we present it for termination, to save the confusion that we've now switched to the forward condition?}\yotamsmall{instrumentation with a derived relation makes sense largely in the forwards, so maybe do everything with forwards?}We now consider a transformation that starts executing the original system only after a step, and show that it is robust for the forwards fence condition.
	The dual version, adding a step before reaching a bad step, is robust for the backwards fence condition.

	The vocabulary is extended with a new variable $\textit{started}$. It is false initially, and becomes true after one step, when the original transition system starts to execute. The system reaches a bad state when it reaches a bad state of the original system and it has already started.
	Thus, $\instr{\Init} = \Init \land \neg\textit{started}$, $\instr{\tr} = \left(\neg \textit{started} \rightarrow \textit{started}' \land \Init' \right) \land \left(\textit{started} \rightarrow \textit{started} \land \tr\right)$\sharon{this means that you can "jump" to any init state, is that what you wanted? less natural than saying $\neg \textit{started} \rightarrow \textit{started}' \wedge id$, also I think  $\left(\textit{started} \rightarrow \tr\right)$ should be $\left(\textit{started} \rightarrow \textit{started}' \wedge \tr\right)$\yotamsmall{fixed bug. yes, I mean to jump. didn't understand the alternative?}}, and $\instr{\Bad} = \Bad \land \textit{started}$.

	The new invariant is $\instr{I} = \neg \textit{started} \lor I$. If $I \in \moncnf{m}{}$, then $\instr{I}$ is equivalent to a $\moncnf{m}{}$ formula (obtained by adding $\neg\textit{started}$ to each clause of $I$).
	The boundary $\boundarypos{\instr{I}}$ consists of all states in $\boundarypos{I}$ where $\textit{started}$ is $\true$, as well as all states in $\neg I$ with $\textit{started}$ being $\false$.\sharon{why doesn't the boundary include $\neg \textit{started} \wedge \neg I$?}\yotamsmall{right, thanks!} The latter are initial states, and if $I$ is forwards $k$-fenced, then the former are $k+1$ reachable. Thus, if $I$ is forwards $k$ fenced then $\instr{I}$ is forwards $k+1$ fenced, thereby guaranteeing efficient inference.

	It is interesting to note that this transformation is robust also with a different initial condition, $\instr{\Init} = \neg\textit{started} \land \psi$ where $\psi \implies \neg I$ (for example, $\psi$ can be the negation of a clause of $I$). Interestingly, in this case, $\instr{I}$ is forwards $k+1$ fenced,\footnote{
		Under this transformation, $\boundarypos{\instr{I}}$ includes, in addition to the aforementioned states, also the states where $\textit{started}$ is $\false$ and $\psi$ does not hold, but they are all $0$-reachable.
	}
	but \emph{not} all states in $\instr{I}$ are reachable when $\psi \not\equiv \true$ (the states satisfying $\neg \textit{started} \land \neg \psi$ are in the invariant but are unreachable)\sharon{isn't it enough to talk about $\neg \textit{started} \land \neg \psi$? these are also in the invariant and unreachable, no? and you don't need $\neg \textit{started} \land \neg \psi$ to be sat}\yotamsmall{corrected, thanks!}, even if this were true in the original system.

	\begin{example}
	\TODO{}
	\end{example}
}

\subsection{Non-Robustness Under Instrumentation}
\subsubsection{Instrumentation by a Derived Relation}
\label{sec:robustness-instrumentation}
Instrumentation by a derived relation introduces new \emph{ghost variables} that have a defined meaning over the program variables, so as to aid program analysis. (See~\Cref{sec:related} for a discussion of the role of instrumentation in inference.)

Formally, instrumentation by a derived relation works as follows:
The vocabulary is extended with a new variable, yielding $\instr{\voc} = \voc \uplus \set{q}$. The transition system is modified to update $q$ according to $\psi$, its intended meaning: $\instr{\Init} \equiv \Init \land q \leftrightarrow \psi$, and $\instr{\tr}$ satisfies $\tr \land q \leftrightarrow \psi \implies \instr{\tr} \land q' \leftrightarrow \psi'$ (reading: if $q$ has the correct interpretation in the pre-state, $\instr{\tr}$ correctly updates it). %
The bad states $\instr{\Bad}$ are all bad under the old definition ($\instr{\Bad} \implies \Bad$), and we expect the projection of $\instr{\Bad}$ to $\voc$ to be exactly $\Bad$.
The example instrumentation in~\Cref{sec:overview:robustness} introduces $q$ to capture $\psi = \oplus_{x_i \in J}{x_i}$; the transitions there do not modify the
truth value of $\psi$, hence in $\instr{\tr}$, they do not modify $q$.

We say that $\instr{\sigma}$ over $\instr{\voc}$ is \emph{consistent} if $\instr{\sigma} \models q \leftrightarrow \psi$. Every reachable state of $(\instr{\Init},\instr{\tr},\instr{\Bad})$ is consistent, but $\boundarypos{I}$ includes inconsistent (thus unreachable) states, and this leads to non-robustness.
\iflong
\begin{lemma}
Let $(\instr{\Init},\instr{\tr},\instr{\Bad})$ be obtained as the instrumentation of $(\Init,\tr,\Bad)$ by the variable $q$ capturing $\psi$.
Let $I$ be over $\voc$ for $(\Init,\tr,\Bad)$.
Then
even if $I$ is an inductive invariant, it is \emph{not} forwards $k$-fenced in $(\instr{\Init},\instr{\tr},\instr{\Bad})$ for any $k \in \mathbb{N}$.
\end{lemma}
\toolong{
\begin{proof}
Let $\sigma^{+}$ over $\voc$ where $\sigma^{+} \in \boundarypos{I}$.\footnote{The boundaries $\boundarypos{I},\boundaryneg{I}$ are never empty, because $I,\neg I$ are both not empty, and the ``crossing point'' when walking from one set to the other by gradually flipping bits is in the boundary.}
We extend $\sigma^{+}$ to $\instr{\sigma^{+}}$ over $\instr{\voc}$ by interpreting $q$ according to whether $\sigma^{+} \models \neg \psi$, that is, $q$ is $\true$ iff $\psi$ is evaluated to $\false$ in $\sigma^{+}$. Note that $\instr{\sigma^{+}} \models I$ but $\instr{\sigma^{+}}$ is ``inconsistent'', namely, $\instr{\sigma^{+}} \not\models q \leftrightarrow \psi$. Therefore, it is unreachable from $\instr{\Init}$ through any number of transitions of $\instr{\tr}$.
It remains to show that $\instr{\sigma^{+}} \in \boundarypos{I}$.
Let $\sigma^{-}$ be a Hamming-neighbor of $\sigma^{+}$ such that $\sigma^{-} \models \neg I$, and extend it to $\instr{\sigma^{-}}$ over $\instr{\voc}$ by interpreting $q$ the same way it is interpreted in $\instr{\sigma^{+}}$. Then $\instr{\sigma^{+}},\instr{\sigma^{-}}$ are Hamming neighbors (w.r.t.\ $\instr{\voc}$), and $\instr{\sigma^{-}} \models \neg I$ as well, since $I$ does not mention $q$ and $\sigma^{-} \models \neg I$. The claim follows.
\end{proof}
}
\else
We show (in the extended version) that even if $I$ ranges over the original vocabulary, it cannot be forwards $k$-fenced in $(\instr{\Init},\instr{\tr},\instr{\Bad})$ for any $k$.
\fi
For example, the invariant of~\Cref{eq:overview-parity-example-inv} is forwards $1$-fenced before the transformation in \Cref{sec:overview:robustness} (see~\Cref{ex:two-sided-fence-overview-example}), but  %
\iflong
it is no longer forwards $k$-fenced after the instrumentation.
\else
not after the instrumentation.
\fi

\yotamsmall{omitting:}
Of course, there is another inductive invariant, $I \land q \leftrightarrow \psi$, that we may set as the target for convergence in our analysis of the algorithm.
Unfortunately, $\boundarypos{I \land q \leftrightarrow \psi} = I \land q \leftrightarrow \psi$ itself, because flipping $q$ moves to a state outside this invariant. Thus, to satisfy our condition, this entire invariant must be $k$-reachable, which is a much stronger requirement than $I$ being $k$-fenced ($I$ must be the least fixed-point and $k$ the diameter). Other $k$-fenced invariants could also exist.

\yotamsmall{lock example in comment}
The discussion in~\Cref{sec:overview:robustness} demonstrates a similar non-robustness of the backwards fence condition.

\subsubsection{Monotonization by Instrumentation}
We have shown that the pre-existing invariants are no longer forwards fenced after an instrumentation.
If we are interested in an invariant that uses the new variable, the effect is rather idiosyncratic and depends on the new invariant.
One class of instrumentations with a clear target for the new invariant is \emph{monotinization by instrumentation}, which is of special importance from the perspective of our results in earlier sections (e.g.~\Cref{thm:monotone-inference-efficient,thm:bshouty-known-efficiency}).
Inspired by the importance of the syntactic structure as displayed by our results, it would seem valuable to ``monotonize'' invariants: introduce new variables capturing the negation of existing variables, and using them to rewrite an invariant $I$ into a \emph{monotone} invariant $\instr{I}$.
For example, suppose that a system admits an inductive invariant $I = (\neg p_1 \lor \neg p_2) \land (p_1 \lor \neg p_3)$, in which the variable $p_1$ appears in $I$ both positively and negatively. We thus introduce a new variable $q$ to be maintained as $\neg p_1$, and write $\instr{I} = (\neg p_1 \lor \neg p_2) \land (\neg q \lor \neg p_3)$, which is (anti)monotone.

Alas, we show that this approach is in general unsuccessful, and does not reduce the inference of fenced invariants from the general case to the monotone case, because $\instr{I}$ is \emph{not} forwards/backwards $k$-fenced in $(\instr{\Init},\instr{\tr},\instr{\Bad})$, even if $I$ is forwards/backwards $k$-fenced in $(\Init,\tr,\Bad)$.
Formally,
let $q$ be a variable aiming to capture the negation of a variable $p \in \voc$. $\instr{I}$ is obtained from $I$ by replacing all occurrences of $\neg p$ by $q$.
We modify the system to maintain the relationship between $q,\neg p$: take $\instr{\Init} = \Init \land p \leftrightarrow \neg q$, and $\instr{\tr} = p \leftrightarrow \neg q \land \tr \land p' \leftrightarrow \neg q'$. Note that $\instr{\tr}$ does not make a step when $q$ is interpreted not as $\neg q$, hence $\instr{I}$ is indeed an inductive invariant.

\sloppy
We first describe the effect on the forwards fence condition when trying to achieve an (anti)monotone invariant (to match~\Cref{thm:monotone-inference-efficient-cnf}).
We show that the transformation above does not produce a forwards $k$-fenced invariant in the case that after replacing occurrences $p$ by $\neg q$, the negative form $\neg p$ must still appear in $\instr{I}$;
roughly, this is when $I$ is not unate w.r.t.\ $p$,
so using $\neg q$ alone cannot be enough
(such as in the example above, in which $\neg p_1$ still appears in $\instr{I}$).
\begin{lemma}
\label{lem:non-robustness-monotnonization}
Let $(\instr{\Init},\instr{\tr},\instr{\Bad})$ be obtained by instrumenting $(\Init,\tr,\Bad)$ with $q$ capturing $\neg p$, and let $I$ be an inductive invariant for $(\Init,\tr,\Bad)$ in negation normal form (negations appear only on atomic variables).
Let $\instr{I}$ be obtained by replacing every occurrence of $p$ in $I$ with $\neg q$.
Suppose that $I$ is not monotone w.r.t.\ the variable $p$:
there is $\sigma$ over $\voc$ such that $\sigma \models I, \sigma[p \mapsto \true] \models \neg I$.
Then $\instr{I}$ is not forwards $k$-fenced in $(\instr{\Init},\instr{\tr},\instr{\Bad})$ for \emph{any} $k \in \mathbb{N}$.
\end{lemma}
\toolong{
\begin{proof}
\yotam{need to recheck the proof...}\sharon{I read it but there are many confusing points, so can't trust me}\yotamsmall{Quis custodiet ipsos custodes?}
The reasoning is similar in spirit to~\Cref{sec:robustness-instrumentation}, but with $\instr{I}$ which is different from $I$.
As before, a state is called consistent if it indeed interprets $q$ as the negation of the interpretation of $p$.
All reachable states are consistent; we shall show that $\boundarypos{\instr{I}}$ includes an inconsistent and hence unreachable state.

First observe that by the construction of $I$, if $\instr{\sigma}$ over $\instr{\voc}$ is consistent, then $\project{\instr{\sigma}}{\voc} \models I$ iff $\instr{\sigma} \models \instr{I}$.

Denote $\sigma^{+} = \sigma$, $\sigma^{-} = \sigma[p \mapsto \true]$, where $\sigma$ is a state as in the premise of the lemma.
Augment them to obtain $\instr{\sigma^{+}},\instr{\sigma^{-}}$ by interpreting $q$ to $\false$ -- the value of $p$ in $\sigma^{+}$; note that $\instr{\sigma^{+}}$ is inconsistent.
In contrast, $\instr{\sigma^{-}}$ is consistent, and thus $\sigma^{-} \models \neg I$, i.e., $\sigma^{-} \not \models I$, implies that $\instr{\sigma^{-}} \not \models \instr{I}$.

Consider now $\instr{\sigma^{+}}$. We know that the consistent $\instr{\sigma^{+}}[q \mapsto \true]$, which is exactly $\sigma^{+}$ with $q$ interpreted consistently, satisfies $I$, because $\sigma^{+}$ does. Since it is consistent, $\instr{\sigma^{+}}[q \mapsto \true] \models \instr{I}$ as well.
But $\instr{I}$ is (anti)monotone w.r.t.\ $q$, so also $\instr{\sigma^{+}} \models \instr{I}$.

Since both states interpret $q$ the same way, $\instr{\sigma^{+}},\instr{\sigma^{-}}$ are Hamming neighbors w.r.t.\ $\instr{\voc}$.
We have shown that $\instr{\sigma^{+}} \models I$ whereas $\instr{\sigma^{-}} \not\models I$ while $\instr{\sigma^{+}}$ is unreachable, and the proof is concluded.
\end{proof}
}

The case of the backwards fence condition is 
\iflong
completely
\fi 
dual when $\instr{\Bad} = \Bad \land p \leftrightarrow \neg q$. If $\instr{\Bad} = \Bad$ itself, a similar non-robustness occurs when the critical $\sigma$ (as in~\Cref{lem:non-robustness-monotnonization}) is not in $\Bad$---even though it may reach $\Bad$ 
\iflong
quickly
\fi 
under $\tr$---because an inconsistent extension of it cannot reach $\Bad$%
\iflong
(the system cannot make any step from an inconsistent state)
\fi
.  
\section{Related Work}
\label{sec:related}

\para{Interpolant generation}
Classically, interpolants are generated from unsatisfiability proofs of SAT/SMT solvers, a procedure which has been studied extensively~\cite[e.g.][]{DBLP:conf/cav/McMillan03,DBLP:journals/tcs/McMillan05,DBLP:conf/fmcad/KroeningW07,DBLP:journals/tocl/CimattiGS10,DBLP:conf/fmcad/McMillan11,DBLP:conf/cav/VizelGM15,DBLP:journals/fmsd/VizelNR15}.
A relatively recent alternative is to sample states that the interpolant must separate~\cite{DBLP:conf/cav/SharmaNA12}. By iteratively sampling and combining the resulting interpoloants, it is possible to compute an interpolant of the (entire) formulas~\cite{DBLP:conf/cav/AlbarghouthiM13,DBLP:conf/cav/DrewsA16,DBLP:conf/lpar/BjornerGKL13,DBLP:conf/hvc/ChocklerIM12}. The interpolation procedure in the algorithm we analyze is from~\citet{DBLP:conf/lpar/BjornerGKL13,DBLP:conf/hvc/ChocklerIM12}, and inspired by IC3/PDR~\cite{ic3,pdr}. It samples from one side---the post image of the previous candidate---and computes an interpolant of the sample and the set of $k$-backwards reachable states. The EPR interpolation method by~\citet{DBLP:conf/cav/DrewsA16} is related, with diagrams replacing cubes as the least abstraction of states. 

\para{Inference and learning}
Ideas from machine learning have inspired many works in invariant inference~\cite[e.g.][]{DBLP:conf/icse/JhaGST10,DBLP:conf/sas/0001GHAN13,DBLP:conf/cav/SharmaNA12,DBLP:conf/esop/0001GHALN13,ICELearning,DBLP:journals/fmsd/SharmaA16,DBLP:conf/pldi/KoenigPIA20}.
A primary concern has been to overcome the ambiguity of counterexamples to induction, or implication counterexamples, compared to the classical learning settings where examples are labeled positive and negative~\cite{ICELearning}.
We resolve this ambiguity through observations of one-sidedness of the algorithm (\Cref{sec:restricted-queries}) or special invariants to which membership can be implemented (\Cref{sec:bshouty-two-sided}).
\citet{DBLP:conf/popl/0001NMR16} target invariants expressed as decision trees over numerical and Boolean attributes, adapting an algorithm by~\citet{DBLP:journals/ml/Quinlan86}. The result in~\Cref{sec:cdnf-inference} applies to decision trees in the propositional setting, and leverages the CDNF algorithm~\cite{DBLP:journals/iandc/Bshouty95}, which admits a bound on the number of hypotheses the learner presents before converging.
The CDNF algorithm~\cite{DBLP:journals/iandc/Bshouty95} has been applied by~\citet{DBLP:journals/mscs/JungKDWY15} to infer quantified invariants through predicate abstraction. They resolve membership queries by over- and under-approximations to some invariants, and use a random guess when the these are not conclusive, potentially leading to failures and restarts. In this work we use bounded model checking to correctly resolve certain queries under the fence condition, and accomplish efficient overall complexity theorems.
\citet{DBLP:journals/pacmpl/FeldmanISS20} have shown that membership and equivalence queries cannot be implemented in general in black-box models of invariant inference. In~\Cref{sec:inference-from-positive-exact} we build on the fence condition to implement some forms of queries. Our results in~\Cref{sec:bshouty-two-sided} generalize the implementability of queries obtained by~\cite{DBLP:journals/pacmpl/FeldmanISS20} the (in three aspects as in~\Cref{sec:intro}). %
The CDNF algorithm has also been applied to generate contextual assumptions in assume-guarantee reasoning~\cite{DBLP:conf/cav/ChenCFTTW10}.

\para{Complexity of invariant inference}
Conjunctive invariants can be inferred in a linear number of SAT calls using Houdini~\cite{DBLP:conf/fm/FlanaganL01}, and similarly for disjunctive invariants using the dual algorithm~\cite{DBLP:conf/cade/LahiriQ09}.
In richer syntactical classes, which are the domain of algorithms such as interpolation~\cite{DBLP:conf/cav/McMillan03} and IC3/PDR~\cite{ic3,pdr}, invariant inference becomes much harder~\cite{DBLP:conf/cade/LahiriQ09}, and is $\NP$-hard with access to a SAT solver even when the problem is restricted to learning monotone CNF/DNF invariants with a linear number of clauses~\cite{DBLP:journals/pacmpl/FeldmanISS20}.
These results place no restrictions on the transition system. %
In this work, the fence condition identifies properties of the transition system and the invariant that evade these hardness results.
In the process of showing a separation between ICE learning and algorithms that use richer queries, \citet{DBLP:journals/pacmpl/FeldmanISS20} also show a polynomial upper bound on the inference of (anti)monotone CNF invariants. Our results go beyond the class they consider (see~\Cref{sec:intro}).
The algorithms we study in this paper fall into the black-box model of extended Hoare queries~\cite{DBLP:journals/pacmpl/FeldmanISS20}. Whether similar efficiency results can also be obtained in the (non-extended) Hoare-query model---without unrollings for BMC---is interesting for future work.

\para{Robustness of program analysis}
The drastic consequences of small program changes on verification tools is sometimes recognized as verification's ``butterfly effect''~\cite{leinotriggerCAV16}.
Many program analysis techniques exhibit brittle behaviors~\cite{DBLP:conf/pldi/LogozzoLFB14,DBLP:conf/vmcai/KarpenkovMW16}. 
This may be in line with the inherent hardness; for instance, the class of programs for which an abstract interpreter is complete is undecidable~\cite{DBLP:conf/popl/GiacobazziLR15}.
A notable exception is the Houdini algorithm, which efficiently finds the strongest conjunctive invariant without any assumptions on the program~\cite{DBLP:conf/fm/FlanaganL01}. In contrast, efficiently inferring invariants from richer syntactic classes is not possible in all programs, because the general case is hard~\cite{DBLP:journals/pacmpl/FeldmanISS20,DBLP:conf/cade/LahiriQ09}, but the fence condition allows efficient inference. As we show, the fence condition is sensitive to some program transformations.

Modifying the program for the sake of verification is emblemed by the practice of ghost code~\cite[see e.g.][]{DBLP:journals/fmsd/FilliatreGP16}.
Improving invariant inference through program transformations has been applied in many different settings with varying notions of the effect on the underlying inference technology, including improving accuracy, enriching the syntactic search space, and simplifying the quantification structure~\cite{TOPLAS:SRW02,DBLP:conf/cav/FeldmanWSS19,DBLP:conf/cav/SharmaDDA11,DBLP:conf/tacas/BorrallerasBLOR17,DBLP:conf/vmcai/Namjoshi07,DBLP:journals/pacmpl/CyphertBKR19}. Perhaps not surprisingly, there is some empirical evidence that instrumentation can also at times have negative impact on the inference algorithm~\cite{DBLP:conf/cav/FeldmanWSS19}.
Our results suggest that such transformations can have a profound effect on the inference algorithm, in the sense that pre-existing invariants can no longer be the cause of guaranteed convergence after instrumentation by a derived relation.

\para{Learning with errors in queries}
Exact learning has also been studied in models where the teacher might not respond ``unknown'' or answer incorrectly; see~\cite{DBLP:journals/tcs/Bshouty18} for a short survey. Usually learning is facilitated by bounding the probability of an error or the number of allowed errors. The fence condition gives rise in~\Cref{sec:restricted-queries} to a different model, which correctly answers  all membership queries to examples that are positive or on the concept's outer-boundary, yet can adversarially classify every other example. As we show, several existing algorithms are applicable in this model. 
\section{Conclusion}
This work applies ideas from exact concept learning for understanding the behavior of interpolation-based invariant inference, achieving the first polynomial complexity result \begin{changebar}with access to a SAT oracle\end{changebar} for such algorithms.
This connection and recent results in inferring quantified invariants and invariants over interesting theories~\cite[e.g.][]{DBLP:journals/jacm/KarbyshevBIRS17,DBLP:conf/sigsoft/GurfinkelSM16,DBLP:conf/popl/0001NMR16,DBLP:conf/cav/FeldmanWSS19,DBLP:conf/pldi/KoenigPIA20,DBLP:conf/oopsla/DilligDLM13,DBLP:journals/fmsd/SharmaA16} suggest that there may be an opportunity to develop exact concept learning for infinite domains, which, to the best of our knowledge, is less explored~\cite[see e.g.][]{DBLP:conf/pods/AbouziedAPHS13,arias200phd}.

\begin{changebar}
The focus of this work in on \emph{worst case} complexity of an interpolation-based invariant inference algorithm, showing conditions that guarantee efficient inference, which is an important step towards a theoretical understanding of invariant inference.
We strived for the weakest possible conditions, so that the results would be applicable to as many programs as possible. We believe that the fence condition is the weakest reasonable condition in this setting---we have not found a graceful way to ensure that the bad examples on the boundary are always avoided (see~\Cref{rem:fence-violation}). We thus believe that to achieve theoretical guarantees for programs that do not satisfy the fence condition, the analysis must go beyond worst case (for example, analyzing convergence with high probability).
\end{changebar}

\begin{acks}
\iflong
\else{\small 
\fi
We thank our shepherd and the anonymous reviewers for comments which improved the paper.
We thank Kalev Alpernas, Nader Bshouty, Shachar Itzhaky, Neil Immerman, Kenneth McMillan, Yishay Mansour, and Oded Padon for insightful discussions, and Sivan Gershon-Feldman for help with~\Cref{fig:fence}.
The research leading to these results has received funding from the
European Research Council under the European Union's Horizon 2020 research and
innovation programme (grant agreement No [759102-SVIS]).
This research was partially supported by the Blavatnik Interdisciplinary Cyber Research Center, Tel Aviv University, Pazy Foundation grant No.\
347853669, the United States-Israel Binational Science Foundation (BSF) grant No. 2016260, and the Israeli Science Foundation (ISF) grant No. 1810/18.
\iflong
\else
}
\fi
\end{acks} 

\clearpage
\bibliography{refs}

\iflong
\clearpage
\appendix

\section{Dual Interpolation With Clause Minimization}
\label{sec:dual-interpolation}
This algorithm computes invariants in CNF and is the dual of~\Cref{alg:itp-termmin}. It is guaranteed to converge to an overapproximation of $I$ when $I$ is forwards $k$-fenced, and converges efficiently when $I$ is a short (anti)monotone CNF formulas (see~\Cref{thm:monotone-inference-efficient-cnf}).
\begin{algorithm}[H]
\caption{Dual-Interpolation by Clause Minimization}
\label{alg:itp-termmin-dual}
\begin{algorithmic}[1]
\begin{footnotesize}
\Procedure{\dualourinterpolationalgname}{$\Init$, $\tr$, $\Bad$, $k$}
	\State $\varphi \gets \neg \Bad$
	\While{$\varphi$ not inductive} $\label{ln:dual-itp:ind-check}$
		\State \textbf{let} $\sigma, \sigma' \models \varphi \land \tr \land \neg \varphi'$ $\label{ln:dual-itp:cti}$
		\If{$\bmc{\tr}{\Init}{k} \cap \set{\sigma'} \neq \emptyset$} $\label{ln:dual-itp:check-no-overapprox}$
			\State \textbf{restart} with larger $k$ $\label{ln:dual-itp-termmin:fail}$
		\EndIf
		\State $d$ $\gets$ $\cube{\sigma'}$ $\label{ln:dual-itp:gen-start}$
		\For{$\ell$ in $d$}
			\If{$\bmc{\tr}{\Init}{k} \cap \left(d \setminus \set{\ell}\right) = \emptyset$} $\label{ln:dual-itp:bmc-gen}$
				\State $d$ $\gets$ $d \setminus \set{\ell}$
			\EndIf
		\EndFor $\label{ln:dual-itp:gen-end}$
		\State $\varphi \gets \varphi \land \neg d$ $\label{ln:dual-itp:add-conjunct}$
	\EndWhile
	\State \Return $I$
\EndProcedure
\end{footnotesize}
\end{algorithmic}
\end{algorithm} %
\section{Algorithms in Exact Concept Learning}
\label{sec:exact-learning-algs}
For an exposition of exact concept learning, see~\Cref{sec:exact-learning-background}.
In the code below,
$\membershipquery{x}$ returns $\true$ iff $x \models \varphi$ where $\varphi$ is the unknown target concept.
$\equivalencequery{H}$ returns $\bot$ if $H \equiv \varphi$ and a differentiating counterexample otherwise.

\begin{algorithm}[H]
\caption{Monotone DNF learning~\cite{DBLP:journals/ml/Angluin87}}
\label{alg:angluin-monotone}
\begin{algorithmic}[1]
\begin{footnotesize}
\Procedure{Learn-Monotone-DNF}{}
	\State $\varphi \gets \false$
	\While{$\sigma'$ $\gets$ $\equivalencequery{\varphi}$ is not $\bot$} $\label{ln:itp-exact:ind-check}$
		\If{$\membershipquery{x} \neq \true$} $\label{ln:itp-exact:check-no-overapprox}$
			\State \textbf{fail} (not monotone) $\label{ln:itp-exact:fail}$
		\EndIf
		\State $d$ $\gets$ $\cubemon{\sigma'}{\vec{0}}$ $\label{ln:itp-exact:gen-start}$
		\For{$\ell$ in $d$}
			\State \textbf{let} $x$ be the state such that $x[p_i]=\true$ iff $p_i \in d \setminus \set{\ell}$
			\If{$\membershipquery{x} = \true$} $\label{ln:itp-exact:bmc-gen}$
				\State $d$ $\gets$ $d \setminus \set{\ell}$
			\EndIf
		\EndFor $\label{ln:itp-exact:gen-end}$
		\State $\varphi \gets \varphi \lor d$ $\label{ln:itp-exact:add-disjunct}$
	\EndWhile
	\State \Return $I$
\EndProcedure
\end{footnotesize}
\end{algorithmic}
\end{algorithm}

\begin{algorithm}[H]
\caption{$\Lambda$-algorithm: exact learning with a known monotone basis~\cite{DBLP:journals/iandc/Bshouty95}}
\label{alg:bshouty-exact-known-basis}
\begin{algorithmic}[1]
\begin{footnotesize}
\State Assuming a known basis $\set{a_1,\ldots,a_t}$ (\Cref{def:monotone-basis})
\Procedure{$\Lambda$-algorithm}{}
	\State $H_1,\ldots,H_t \gets \false$
	\While{$\sigma'$ $\gets$ $\equivalencequery{\bigwedge_{i=1}^{t}{H_i}}$ is not $\bot$} $\label{$ln:bhsouty-exact:positive-example}$
		\For{$i=1,\ldots,t$}
			\If{$\sigma' \not\models H_i$}
				\State $d$ $\gets$ \Call{monotone-generalize-membership}{$\sigma'$, $a_i$, $\Init$, $\tr$, $k$}
				\State $H_i \gets H_i \lor d$
			\EndIf
		\EndFor
	\EndWhile
	\State \Return $\bigwedge_{i=1}^{t}{H_i}$
\EndProcedure

\Procedure{monotone-generalize-membership}{$\sigma$, $a$, $\tr$, $\Bad$, $k$}
	\If{$\membershipquery{x} \neq \true$} $\label{ln:bshouty-exact:check-no-overapprox}$
		\State \textbf{fail} (not a basis) $\label{ln:bshouty-exact:fail}$
	\EndIf
	\State $v \gets \sigma$
    \State walked $\gets$ $\true$
	\While{walked}
		\State walked $\gets$ $\false$
		\For{$j=1,\ldots,n$ such that $x[p_j] \neq v[p_j]$}
			\State $x \gets v[p_j \mapsto a[p_j]]$
			\If{$\membershipquery{x} = \true$} $\label{ln:bshouty-exact:check-membership}$
				\State $v \gets$ x
				\State walked $\gets$ $\true$
			\EndIf
		\EndFor
	\EndWhile
	\State \Return{$\cubemon{v}{a}$}
\EndProcedure
\end{footnotesize}
\end{algorithmic}
\end{algorithm} 
\fi

\end{document}